\documentclass[a4paper, 12pt]{article}

\pdfoutput=1

\usepackage{color}
\usepackage{graphicx,epic} 
\usepackage{amsmath}
\usepackage{amsfonts}
\usepackage{amssymb}
\usepackage{amsthm}
\usepackage{epstopdf}
\usepackage{algorithmic}
\usepackage{setspace}
\usepackage[cm]{fullpage}
\usepackage{xcolor}
\usepackage{lineno} 

\newtheorem{remark}{\bf Remark}

\newtheorem{lemma}{\bf Lemma}

\newtheorem{theorem}{\bf Theorem}

\newcommand{\mbf}[1]{\mbox{\boldmath$#1$}}

\newcommand{\sbv}{\mbox{\boldmath{\scriptsize $v$}}}
\newcommand{\sbs}{\mbox{\boldmath{\scriptsize $\sigma$}}}

\newcommand\crule[3][black]{\textcolor{#1}{\rule{#2}{#3}}}

\numberwithin{equation}{section}
\numberwithin{definition}{section}
\numberwithin{lemma}{section}
\numberwithin{remark}{section}
\numberwithin{corollary}{section}
\numberwithin{theorem}{section}
\numberwithin{example}{section}

\begin{document}
	
\title{Optimization and variational principles for the shear strength reduction method}
\author{Stanislav Sysala$^1$\footnote{corresponding author, email: \texttt{stanislav.sysala@ugn.cas.cz}}, Eva Hrube\v sov\'a$^{2,1}$, Zden\v ek Michalec$^{1}$, Franz Tschuchnigg$^{3}$ \\ \\
\small$^1$Institute of Geonics of the Czech Academy of Sciences, Ostrava, Czech Republic\\
\small$^2$V\v SB - Technical University of Ostrava, Faculty of Civil Engineering, Ostrava, Czech Republic\\
\small$^3$Institute of Soil Mechanics, Foundation Engineering and Computational Geotechnics,\\ 
\small Graz University of Technology, Graz, Austria}

\maketitle

\begin{abstract}
This paper is focused on the definition, analysis and numerical solution of a new optimization variant (OPT) of the shear strength reduction (SSR) problem with applications to slope stability problems. This new variant is derived on the basis of recent results by Tschuchnigg et al. 2015, where limit analysis and a modified Davis approach were used for approximation of the standard SSR method. The OPT-SSR method computes the factor of safety without performing an elasto-plastic analysis, similarly as in limit analysis. It is shown that this optimization problem is well-defined. Next, the  duality between the static and kinematic principles of OPT-SSR is derived. For the numerical solution, a regularization method is introduced and analyzed. This method is combined with the finite element method, mesh adaptivity and a damped Newton method. In-house codes (Matlab) are used for the implementation of this solution concept. Finally, two slope stability problems are considered, one of which follows from analysis of a real slope. The softwares packages Plaxis and Comsol Multiphysics are used for comparison of the results.
\end{abstract}

{\bf Keywords:}  slope stability, shear strength reduction method, convex optimization, static and kinematic principles, regularization, finite elements and mesh adaptivity



\section{Introduction}
\label{sec_intro}

This paper deals with slope stability assessment, which includes the determination of the factor of safety (FoS) and the estimation of failure zones (slip surfaces) for a critical state of the slope. Limit equilibrium (LE), shear strength reduction (SSR) or limit analysis (LA) can be used for the determination of FoS. The methods arise from elastic-perfectly plastic models containing mainly the Mohr-Coulomb yield criterion.

The LE method is based on predefined failure zones, see, for example, \cite{D96}, \cite{Y06}. It does not necessarily require numerical computation, and neither stress equilibrium at every point in the domain around the slope. Due to these facts, this method is simple and widely used in geotechnical practice. On the other hand, accuracy of the solution cannot be easily verified, especially if anisotropic or inhomogeneous materials are considered or if a complex geometry is defined.

The SSR method \cite{ZHL1975,BB91},\\ \cite{DRD99,GL99} has been suggested mainly for the slope stability assessment. It is a conventional method based on a displacement variant of the finite element method (FEM) and on reduction of strength parameters defining the Mohr-Coulomb model. It has also been implemented within some commercial softwares like Plaxis \cite{BSE11} or Middas GTS NX. Strong dependence on the finite element mesh density or even a nonunique determination of FoS can occur in SSR with a non-associated plastic flow rule, see \cite{TSSLR2015,TSS2015b}.

LA is a universal method that can be used for various stability problems, not only for the ones in geotechnical practice. FoS is derived from a critical (limit) value of the load factor. Originally, this method was purely analytical, see, for example, \cite{CL90}, \cite{MD09}. Now, it is rather a numerical method based on optimization, duality between kinematic and static principles and performed within the framework of FEM \cite{Ch96,Sl13,Y06,HRS19}. LA is supported by mathematical and numerical analyses, see \cite{T85},\\ \cite{Ch96,HRS15,RSH18,HRS19} and has been implemented, for example, within the software OPTUMG2 and OPTUMG3 \cite{KLK16}. In slope stability analyses, however, the FoS is defined according to  strength parameters. Therefore, an iterative modification of LA was suggested, see e.g. \cite{Sl13}. In addition, there is also no rigorous solution in the case of non-associated plasticity.

The usage of the non-associated plastic flow rule is supported by laboratory experiments and enables to control the inelastic volume changes of compacted dense granular soil materials or overconsolidated fine-grained soils material subjected to shearing, see, for example, \cite{VB84, Sch05}. On the other hand, mathematical theory of non-associated elastic-plastic problems is missing or at least incomplete. Especially, standard implicit discretizations of time (pseudo-time) variables lead to problematic numerical behavior. The drawbacks of non-associated models can be suppressed by variational approaches based on theory of bipotentials \cite{HFS03, HQS17} or semi-implicit time schemes \cite{KKLS2012}. Within the LA method, Davis \cite{D68} suggested to modify the strength parameters and consequently approximate the non-associated model by the associated one. In\\ \cite{TSSLR2015,TSS2015b}, the Davis approach has been modified for purposes of the SSR method. The modification leads to an iterative solution scheme based on LA that was originally suggested in \cite{Sl13}. We also refer to recent papers \cite{TSS2015c, OTS18} for comparison of the standard and the modified SSR method.

The drawback of the modified SSR method developed in \cite{TSS2015b} is related to the fact that the FoS is defined iteratively. The aim of this paper is to propose a direct optimization variant of the strength reduction method (OPT-SSR). This proposed procedure defines FoS without performing an elasto-plastic analysis. Consequently, one can study OPT-SSR regardless its space discretization and thus, this approach has a potential to be completed by a rigorous mathematical theory (as in LA). We show in the following some basic properties of the optimization problem and derive the duality between static and kinematic settings of this problem.

The mentioned iterative scheme from \cite{Sl13, TSS2015b} can be used for the solution of the OPT-SSR problem. However, the LA problem has to be solved in each iteration of this scheme. In addition, mesh adaptivity usually completes LA solvers (see \cite{Sl13, HRS19}) and its repetitive construction is expensive. Therefore, we propose a regularization method for the solution of OPT-SSR which is more straightforward than the LA approach. The used regularization is inspired by recent papers \cite{SHHC15,CHKS15,HRS15}, \cite{HRS16, RSH18, HRS19}, and\\ \cite{RS20}. Convergence with respect to the regularization parameter is analyzed in order to relate the OPT-SSR problem with its regularized counterpart. The regularization enables to solve the problem with standard finite element methods and with the damped Newton method as suggested in \cite{S12}. Mesh adaptivity is also used to compute more accurate results.

The problem is implemented within in-house Matlab codes. These codes (based on elastic-plastic solvers) have been systematically developed and described in \cite{SCKKZB16,SCL17,CSV19}. Some of the codes are available for download \cite{CSV18}. For comparison of the results with standard approaches, the commercial softwares Plaxis and Comsol Multiphysics are used. 

The paper is organized as follows. In Section \ref{sec_preliminaries}, we introduce preliminaries related to the standard SSR method. In Section \ref{sec_SSR}, the OPT-SSR problem is introduced for associated plasticity. Section \ref{sec_nonassoc} is devoted to the extension of OPT-SSR to non-associated plasticity. The extended problem is formulated for three different Davis modifications as suggested in \cite{TSS2015b} and the corresponding safety factors are compared. In Section \ref{sec_LA}, the LA approach for the solution of OPT-SSR is introduced. It enables to relate the OPT-SSR method to the approaches from \cite{Sl13, TSS2015b}. In Section \ref{sec_duality}, variational principles, duality and the kinematic approaches to the OPT-SSR method are presented. The regularization method, which is built on the variational principals, is introduced and analyzed in Section \ref{sec_regularization}. Numerical examples illustrating the efficiency of the suggested numerical methods are presented in Section \ref{sec_examples}. Concluding remarks are given in Section \ref{sec_conclusion}. The appendix contains a closed form of a regularized dissipative function for the Mohr-Coulomb yield criterion.


\section{Preliminaries to the standard SSR method}
\label{sec_preliminaries}

The standard SSR method is based on the elastic-perfectly plastic problem including the Mohr-Coulomb yield criterion. For the complete definition of this problem, we refer to \cite{NPO08, SCL17}. Such a model includes the elastic material parameters (Young's modulus and Poisson's ratio) and the following strength parameters: the effective cohesion ($c'$), the effective friction angle ($\phi'$), and the dilatancy angle ($\psi'$). It is assumed that $\psi'\leq\phi'$. In case of $\psi'=\phi'$, we arrive at an associated flow rule.

The SSR method is based on the reduction of the strength parameters $c'$, $\phi'$ and $\psi'$:
\begin{equation}
c_\lambda:=\frac{c'}{\lambda},\quad  \phi_\lambda:=\arctan\frac{\tan\phi'}{\lambda},\quad \psi_\lambda:=\arctan\frac{\tan\psi'}{\lambda},
\label{reduction}
\end{equation}
where $\lambda>0$ is the reduction parameter. Alternatively, one can use the following formula for $\psi_\lambda$ (see also \cite{TSS2015b}):
\begin{equation}
\psi_\lambda:=\psi\; \mbox{ until } \psi<\phi_\lambda,\; \mbox{ then }\psi_\lambda:=\phi_\lambda.
\label{psi_lambda}
\end{equation}
FoS for the SSR method is defined as a maximum of $\lambda$ for which the elastic-perfectly plastic problem has a solution with respect to the parameters $c_\lambda$, $\phi_\lambda$, and $\psi_\lambda$. This definition is from the mathematical point of view rather formal, because the solvability of the elasto-plastic problem requires to introduce convenient functional spaces and a weak form of the problem (see, e.g., \cite{HR12}). Such an analysis is problematic for the case of the non-associated plasticity.

Finally, we introduce an appropriate form of the Mohr-Coulomb yield criterion convenient for the analysis of the OPT-SSR method presented below. Arising from the well-known formulas
\begin{equation}
\cos\phi'=\frac{1}{\sqrt{1+\tan^2\phi'}},\quad \sin\phi'=\frac{\tan\phi'}{\sqrt{1+\tan^2\phi'}},
\label{tan}
\end{equation}
we arrive at the following form of the criterion for $c_\lambda$ and $\phi_\lambda$:
\begin{equation}
(\sigma_1-\sigma_3)\sqrt{1+\tan^2\phi_\lambda}+(\sigma_1+\sigma_3)\tan\phi_\lambda-2c_\lambda\leq 0,
\label{MC2}
\end{equation}
where $\sigma_1$ and $\sigma_3$ denote the maximal and minimal principle effective stresses of the Cauchy stress tensor $\mbf\sigma$, respectively. Notice that the standard mechanical sign convention has been used. Using \eqref{reduction} and multiplying \eqref{MC2} by $\lambda$, one can rewrite this inequality into the form $\Phi(\lambda;\mbf\sigma)\leq0$, where
\begin{equation}
\Phi(\lambda;\mbf\sigma):=(\sigma_1-\sigma_3)\sqrt{\lambda^2+\tan^2\phi'}+(\sigma_1+\sigma_3)\tan\phi'-2c'.
\label{Phi_lambda}
\end{equation}

\section{The OPT-SSR method in associated plasticity}
\label{sec_SSR}

Inspired by the LA method, we introduce an optimization variant of the SSR method (OPT-SSR) for the associated model with $\phi'=\psi'$ in the following form:

\smallskip\noindent
\textit{$\lambda^*_{ass}=$ supremum of $\lambda\geq0$ subject to}
\begin{equation}
\left.\begin{array}{c}
-\mathrm{div}\, \mbf\sigma=\mbf F\;\mbox{ in }\Omega,\;\;\mbf\sigma\mbf n=\mbf f\;\mbox{ on }\partial\Omega_f,\\[2mm]
\Phi(\lambda;\mbf\sigma)\leq0\; \mbox{ in }\Omega.
\end{array}\right\}
\label{SSR_FoS2}
\end{equation}
Here, $\lambda^*_{ass}$ denotes FoS for the associated OPT-SSR method, $\Omega$ is a bounded domain in 2D and 3D representing an investigated body, $\mbf F$ is a volume force (e.g. the weight of the body), $\mbf f$ is a prescribed surface force acting on the part $\partial \Omega_f$ of the boundary $\partial\Omega$, $\mbf n$ denotes the outward unit normal to the boundary $\partial\Omega$, and the function $\Phi(\lambda;\mbf\sigma)$ is defined by \eqref{Phi_lambda}. The constraints on the first and second lines of \eqref{SSR_FoS2} represent statically and plastically admissible stress fields, respectively. According to the literature on convex analysis \cite{ET74,T85,Ch96}, we rather use the supremum than the maximum in this definition, because for the critical value $\lambda^*_{ass}$, the admissible stress $\mbf\sigma$ satisfying \eqref{SSR_FoS2} does not need to exist on functional spaces. Although the definition admits the case $\lambda^*_{ass}=+\infty$, one can expect that $\lambda^*_{ass}$ is finite in geotechnical boundary value problems.

The following statement implies that \eqref{SSR_FoS2} holds for any $\lambda\geq0$ such that $\lambda<\lambda^*_{ass}$. Without this basic property, it would be very difficult to find $\lambda^*_{ass}$. More advanced analysis of the OPT-SSR problem can be found in Section \ref{sec_duality}.

\begin{lemma}
If \eqref{SSR_FoS2} is satisfied for some $\lambda:=\overline\lambda>0$ then \eqref{SSR_FoS2} holds for any $\lambda<\overline\lambda$.
\label{lem_well_def}
\end{lemma}
\begin{proof}
For any $\mbf\sigma$ fixed, $\sigma_1-\sigma_3\geq0$ and thus the function $\lambda\mapsto\Phi(\lambda;\mbf\sigma)$ is nondecreasing. Hence, if there exists $\mbf\sigma$ such that \eqref{SSR_FoS2} holds for some $\lambda:=\overline\lambda>0$
then for any $\lambda<\overline\lambda$, we have
\begin{equation}
\Phi\big(\lambda;\mbf\sigma\big)\leq\Phi\big(\bar\lambda;\mbf\sigma\big)\leq0\; \mbox{ in }\Omega.
\end{equation} 
This implies the statement of the lemma.
\end{proof}

\begin{remark}
\emph{Let us recall that the Mohr-Coulomb yield surface is the pyramid aligned with the hydrostatic axis, see, for example, \cite{NPO08}. From the inequality $\Phi(\lambda;\mbf\sigma)\leq0$ and \eqref{Phi_lambda}, it is possible to see that the apex of this pyramid is independent of $\lambda$. By reducing the strength parameters (i.e., by enlarging $\lambda$) the slope of the Mohr-Coulomb pyramid is reduced. For $\lambda\rightarrow0$, the pyramid varies to a half-space.}
\end{remark}

\section{The OPT-SSR method in non-associated plasticity}
\label{sec_nonassoc}

The aim of this section is to extend the OPT-SSR method from Section \ref{sec_SSR} to the non-associated model where $\psi'<\phi'$. To this end, we use the Davis approach and its modifications suggested for the SSR and LA methods, see \cite{TSS2015b}. In particular, three different approaches denoted as Davis A, Davis B and Davis C are distinguished in \cite{TSS2015b} and are considered in this paper. 

In general, we propose the reduction of the strength parameters $c'$ and $\tan\phi'$ by the following scheme:
\begin{equation}
\frac{c'}{q(\lambda)},\quad \frac{\tan\phi'}{q(\lambda)},\quad q(\lambda):=\frac{\lambda}{\beta(\lambda)}.
\label{Davis}
\end{equation}
Different definitions of the function $\beta$ for the Davis A-C approaches (according their advantages and disadvantages) are introduced in \cite{TSS2015b}. Therefore, we present directly the corresponding functions $q$, see Sections \ref{subsec_Davis_A}--\ref{subsec_Davis_C} listed below. The extended OPT-SSR problem is in the following abstract form:

\smallskip\noindent
\textit{$\omega^*=$ supremum of $\lambda\geq0$ subject to}
\begin{equation}
\left.\begin{array}{c}
-\mathrm{div}\, \mbf\sigma=\mbf F\;\mbox{ in }\Omega,\;\;\mbf\sigma\mbf n=\mbf f\;\mbox{ on }\partial\Omega_f,\\[2mm]
\Phi\big(q(\lambda);\mbf\sigma\big)\leq0\; \mbox{ in }\Omega,
\end{array}\right\}
\label{SSR_FoS_N}
\end{equation}
where $\omega^*$ denotes FoS and the function $\Phi$ is defined by \eqref{Phi_lambda}, that is,
\begin{equation}
\Phi(q(\lambda);\mbf\sigma):=(\sigma_1-\sigma_3)\sqrt{q^2(\lambda)+\tan^2\phi'}+(\sigma_1+\sigma_3)\tan\phi'-2c'.
\label{Phi_lambda2}
\end{equation}
Setting $q:=q_{ass}$ where $q_{ass}(\lambda)=\lambda$, we arrive at the associated OPT-SSR problem studied in Section \ref{sec_SSR} and thus $\omega^*=\lambda^*_{ass}$ for this case. Next, we have the following extension of Lemma \ref{lem_well_def}.
\begin{lemma}
Let the function $q$ be non-decreasing. If \eqref{SSR_FoS_N} is satisfied for some $\lambda:=\overline\lambda>0$ then \eqref{SSR_FoS_N} holds for any $\lambda<\overline\lambda$.
\end{lemma}
One can see that it is important to be $q$ non-decreasing. Besides, it is reasonable to assume that the function $q$ is also non-negative and continuous. Now, we shall introduce this function to the Davis A-C approaches and verify the above mentioned properties. To distinguish the particular Davis approaches, we use the notation $q_A$, $q_B$, and $q_C$ instead of $q$.

\subsection{The approach Davis A}
\label{subsec_Davis_A}

We set 
\begin{equation}
q_A(\lambda)=\lambda\,\frac{1-\sin\psi'\sin\phi'}{\cos\psi'\cos\phi'}.
\label{q_A}
\end{equation}
Clearly, the function $q_A$ is increasing, non-negative and continuous. Using the formulas \eqref{tan}, one can also write:
\begin{equation}
q_A(\lambda)=\lambda\left[\sqrt{(1+\tan^2\psi')(1+\tan^2\phi')}-\tan\psi'\tan\phi'\right].
\label{q_A2}
\end{equation}
We denote FoS for the function $q_A$ by $\lambda^*_A$, that is, $\omega^*=\lambda^*_A$ for $q:=q_A$.

\subsection{The approach Davis B}
\label{subsec_Davis_B}

We set
\begin{equation}
q_B(\lambda)=\lambda\,\frac{1-\sin\psi_\lambda\sin\phi_\lambda}{\cos\psi_\lambda\cos\phi_\lambda}
\label{B}
\end{equation}
where the functions $\lambda\mapsto\phi_\lambda$ and $\lambda\mapsto\psi_\lambda$ are defined by $\eqref{reduction}$, that is,
\begin{equation}
\tan\phi_\lambda:=\frac{\tan\phi'}{\lambda},\quad \tan\psi_\lambda:=\frac{\tan\psi'}{\lambda}.
\label{reduction_B}
\end{equation}
Using the formulas \eqref{tan} and \eqref{reduction_B}, we derive subsequently:
\begin{equation}
q_B(\lambda)=\frac{1}{\lambda}\left[\sqrt{(\lambda^2+\tan^2\psi')(\lambda^2+\tan^2\phi')}-\tan\psi'\tan\phi'\right],
\label{q_B2}
\end{equation}
\begin{equation}
\frac{dq_B(\lambda)}{d\lambda}=\frac{1+\frac{1}{\lambda^2}[\sqrt{(\lambda^2+\tan^2\psi')(\lambda^2+\tan^2\phi')}-\tan\psi'\tan\phi']}{\sqrt{(\lambda^2+\tan^2\psi')(\lambda^2+\tan^2\phi')}}>0.
\end{equation}
It implies that the function $q_B$ is increasing, non-negative and continuous. We denote FoS for the function $q_B$ by $\lambda^*_B$, that is, $\omega^*=\lambda^*_B$ for $q:=q_B$. Due to the fact that the difference between $\phi_\lambda$ and $\psi_\lambda$ defines the amount of non-associativity, the methods Davis B is considered to be more appropriate compared to Davis A and also Davis C (mentioned below).

\subsection{The approach Davis C}
\label{subsec_Davis_C}

We set
\begin{equation}
q_C(\lambda)=\left\{\begin{array}{cc}
\lambda\,\frac{1-\sin\psi'\sin\phi_\lambda}{\cos\psi'\cos\phi_\lambda}, & \mbox{if } \phi_\lambda\geq\psi',\\[2mm]
\lambda, & \mbox{if } \phi_\lambda\leq\psi',
\end{array}\right.
\label{b}
\end{equation}
where the function $\lambda\mapsto\phi_\lambda$ is defined by $\eqref{reduction_B}$. Using the formulas \eqref{tan}, we derive:
\begin{equation}
	q_C(\lambda)=\left\{\begin{array}{cc}
	\sqrt{(\lambda^2+\tan^2\phi')(1+\tan^2\psi')}-\tan\phi'\tan\psi', & \mbox{if } \tan\phi'\geq\lambda\tan\psi',\\[2mm]
	\lambda, & \mbox{if } \tan\phi'\leq\lambda\tan\psi',
	\end{array}\right.
	\label{q_C2}
\end{equation}
Hence, the function $q_C$ is non-decreasing, non-negative and continuous. We denote FoS for the function $q_C$ by $\lambda^*_C$, that is, $\omega^*=\lambda^*_C$ for $q:=q_C$. 

\subsection{Comparison of the safety factors $\lambda^*_{ass}$, $\lambda^*_A$, $\lambda^*_B$, and $\lambda^*_C$}

It will be shown that the presented OPT-SSR approach enables to compare analytically the values $\lambda^*_{ass}$, $\lambda^*_A$, $\lambda^*_B$, and $\lambda^*_C$ of FoS for the associative model and Davis A-C models, respectively. First, we compare the corresponding functions $q_{ass}$, $q_A$, $q_B$, and $q_C$.

\begin{lemma}
The following statements hold:
\begin{enumerate}
\item $q_A\geq q_{ass}$, $q_B\geq q_{ass}$, $q_C\geq q_{ass}$;
\item $q_A(1)=q_B(1)=q_C(1)$;
\item $q_A(\lambda)\geq q_B(\lambda)\geq q_C(\lambda)\;$ for any $\;\lambda\geq 1$;
\item $q_C(\lambda)\geq q_B(\lambda)\geq q_A(\lambda)\;$ for any $\;\lambda\leq 1$.
\end{enumerate}
\label{lem_q_compare}
\end{lemma}

\begin{proof}
We use the following simplifying notation: $a:=\tan\phi$ and $b:=\tan\psi$, that is $0\leq b\leq a$. From \eqref{q_A2}, \eqref{q_B2}, and \eqref{q_C2}, one can derive the following equalities:
\begin{equation}
\frac{q_A(\lambda)}{\lambda}=1+\frac{(a-b)^2}{\sqrt{(1+a^2)(1+b^2)}+1+ab},
\label{q_A4}
\end{equation}
\begin{equation}
\frac{q_B(\lambda)}{\lambda}=1+\frac{(a-b)^2}{\sqrt{(\lambda^2+a^2)(\lambda^2+b^2)}+\lambda^2+ab},
\end{equation}
\begin{equation}
\frac{q_C(\lambda)}{\lambda}=\left\{\begin{array}{cc}
1+\frac{(a-\lambda b)^2}{\lambda\left[\sqrt{(\lambda^2+a^2)(\lambda^2+b^2)}+\lambda^2+ab\right]}, & \mbox{if } \lambda \leq a/b \\[2mm]
1, & \mbox{if } \lambda \geq a/b.
\end{array}\right.
\label{C4}
\end{equation}
Hence, it is readily seen that the first two statements hold. The relations between $q_A$ and $q_B$ also hold. To relate $q_B$ and $q_C$, we use the following inequalities:
\begin{equation*}
\left.\begin{array}{cl}
(a-\lambda b)^2\leq (a-b)^2,& \mbox{if } 1\leq\lambda \leq a/b,\\
(a-\lambda b)^2\geq (a-b)^2,& \mbox{if } 1\geq\lambda,
\end{array}
\right.
\end{equation*}
and
\begin{equation*}
\left.\begin{array}{cl}
\lambda\left[\sqrt{(\lambda^2+a^2)(\lambda^2+b^2)}+\lambda^2+ab\right]\geq \sqrt{(\lambda^2+a^2)(\lambda^2+b^2)}+\lambda^2+ab,& \mbox{if } 1\leq\lambda,\\[2mm]
\lambda\left[\sqrt{(\lambda^2+a^2)(\lambda^2+b^2)}+\lambda^2+ab\right]\leq \sqrt{(\lambda^2+a^2)(\lambda^2+b^2)}+\lambda^2+ab,& \mbox{if } 1\geq\lambda.
\end{array}
\right.
\end{equation*}
\end{proof}

From Lemma \ref{lem_q_compare}, one can easily derive the following findings which are in accordance with numerical observations presented in \cite{TSS2015b, OTS18}.
\begin{theorem}
	The following statements hold:
	\begin{enumerate}
		\item $\lambda^*_A\leq \lambda^*_{ass}$, $\lambda^*_B\leq \lambda^*_{ass}$, $\lambda^*_C\leq \lambda^*_{ass}$;
		\item either $\;1\leq\lambda^*_A\leq\lambda^*_B\leq\lambda^*_C\;$ or $\;1\geq\lambda^*_A\geq\lambda^*_B\geq\lambda^*_C$.
		\item If one of the values $\lambda^*_A,\,\lambda^*_B,\,\lambda^*_C$ is equal to one then the same holds for the remaining values.
	\end{enumerate}
\label{theorem_FoS}
\end{theorem}
\begin{proof}
Let $\lambda<\lambda^*_A$. Then the constraints in \eqref{SSR_FoS_N} are satisfied for $\lambda$ and $q:=q_A$. Since $q_A\geq q_{ass}$, we have
$$0\geq\Phi(q_A(\lambda),\mbf\sigma)\geq\Phi(q_{ass}(\lambda),\mbf\sigma).$$
It means that the constraints in \eqref{SSR_FoS_N} are also satisfied for $\lambda$ and $q:=q_{ass}$. This implies $\lambda^*_A\leq \lambda^*_{ass}$. Analogously, one can prove $\lambda^*_B\leq \lambda^*_{ass}$ and $\lambda^*_C\leq \lambda^*_{ass}$, and thus the first statement holds.

Let one of the values $\lambda^*_A$, $\lambda^*_B$, $\lambda^*_C$ be greater than one. Then, using the equalities $q_A(1)=q_B(1)=q_C(1)$, the constraints in \eqref{SSR_FoS_N} are satisfied for $\lambda:=1$ and $q:=q_A$ or $q:=q_B$ or $q:=q_C$. This implies $\lambda^*_A\geq 1$, $\lambda^*_B\geq 1$, $\lambda^*_C\geq 1$. From the inequalities $q_A(\lambda)\geq q_B(\lambda)\geq q_C(\lambda)$ which hold for any $\lambda\geq 1$, we consequently derive $1\leq\lambda^*_A\leq\lambda^*_B\leq\lambda^*_C$ similarly as in the first part of the proof.

Let $\lambda^*_C\leq1$. Then, for any $\lambda<1$, we have $q_C(\lambda)\geq q_B(\lambda)\geq q_A(\lambda)$ implying $\lambda^*_A\geq\lambda^*_B\geq\lambda^*_C$. In addition, the inequality $\lambda^*_A\leq1$ must hold as a consequence of the previous part of the proof. Therefore, the second statement holds.

Third statement is a direct consequence of the second statement.
\end{proof}

Let us note that if $\lambda^*_C\geq\tan\phi'/\tan\psi'$ then $\lambda^*_C=\lambda^*_{ass}$. If $\psi'=0^\circ$ then $\lambda^*_B=\lambda^*_C$. In this case, the Davis B and Davis C approaches coincide. 

It is also important to note that the suggested OPT-SSR problem is not limited to the choice of $q_A$, $q_B$, and $q_C$ of the function $q$. The choice of the function $q$ can also be optimized, for example, by inverse analysis.

\section{Iterative limit analysis for the solution of OPT-SSR}
\label{sec_LA}

The aim of this section is to relate the iterative LA solution scheme from \cite{Sl13} to the OPT-SSR problem. Consider a fixed value of $\lambda\geq0$ and the corresponding reduction parameter $q(\lambda)$ for a given function $q$. With respect to this parameter, we define the LA problem as follows:

\smallskip\noindent
\textit{$\ell(\lambda)=$ supremum of $\zeta\geq0$ subject to}
\begin{equation}
\left.\begin{array}{c}
-\mathrm{div}\, \mbf\sigma=\zeta\mbf F\;\mbox{ in }\Omega,\;\;\mbf\sigma\mbf n=\zeta\mbf f\;\mbox{ on }\partial\Omega_f,\\[2mm]
\Phi(q(\lambda),\mbf\sigma)\leq0\; \mbox{ in }\Omega.
\end{array}\right\}
\label{LA_FoS3}
\end{equation}
The value $\ell(\lambda)$ defines the safety factor of the LA problem depending on $\lambda$.
We shall discuss properties of the corresponding function $\ell$. 
\begin{lemma}
	Let the function $q$ be non-decreasing. Then the function $\ell$ is non-increasing.
\end{lemma}	
\begin{proof}
	Let $\lambda_1\leq\lambda_2$ be two arbitrary values. To prove $\ell(\lambda_1)\geq\ell(\lambda_2)$ it suffices to show that the following implication holds for any $\zeta\geq 0$: if $\zeta<\ell(\lambda_2)$ then $\zeta\leq\ell(\lambda_1)$. Let us suppose that $\zeta<\ell(\lambda_2)$. Then there exists $\mbf\sigma$ satisfying \eqref{LA_FoS3} for $\lambda:=\lambda_2$. Since the function $q$ is non-decreasing by the assumption, we have
	$$\Phi(q(\lambda_1),\mbf\sigma)\leq\Phi(q(\lambda_2),\mbf\sigma)\leq0.$$
	Hence, $\mbf\sigma$ satisfies \eqref{LA_FoS3} also for $\lambda:=\lambda_1$. Therefore, $\zeta\leq\ell(\lambda_1)$.
\end{proof}
If the function $q$ is continuous, then one can also expect that $\ell$ is continuous. By comparison of the constraints \eqref{SSR_FoS2} and \eqref{LA_FoS3}, we derive that the safety factor $\omega^*$ of the OPT-SSR problem introduced in Section \ref{sec_nonassoc} is a solution of the following equation: 
\begin{equation}
\ell(\omega^*)=1.
\label{l*}
\end{equation}
If we solve this equation iteratively, we arrive, for example, at the algorithm introduced in \cite{Sl13}. Consequently, the safety factors $\lambda^*_A$, $\lambda^*_B$ and $\lambda^*_C$ presented above should be very close to the safety factors computed with finite element limit analysis (FELA) Davis A-C as presented in \cite{TSS2015b}.

However, repetitive solution of the LA problem is expensive, especially, if the LA is combined with mesh adaptivity which improves the quality of the computed results significantly. Therefore, we shall derive in Section \ref{sec_regularization} a more straightforward method for solution of OPT-SSR.

\section{Variational principles, duality and kinematic approaches}
\label{sec_duality}

We consider the abstract OPT-SSR problem introduced in Section \ref{sec_nonassoc}:

\smallskip\noindent
\textit{$\omega^*=$ supremum of $\lambda\geq0$ subject to}
\begin{equation}
\left.\begin{array}{c}
-\mathrm{div}\, \mbf\sigma=\mbf F\;\mbox{ in }\Omega,\;\;\mbf\sigma\mbf n=\mbf f\;\mbox{ on }\partial\Omega_f,\\[2mm]
\Phi\big(q(\lambda);\mbf\sigma\big)\leq0\; \mbox{ in }\Omega,
\end{array}\right\}
\label{SSR_FoS}
\end{equation}
where $q$ is an increasing, continuous and nonnegative function. 
This problem can be interpreted as the static principle of the OPT-SSR method. The aim of this section is to derive the corresponding kinematic principle, which will be used for the numerical solution of the problem. Since a similar derivation is known in the limit analysis problem \cite{T85,Ch96,HRS19}, some technical details are skipped, for the sake of brevity. 

We introduce the following functional spaces: 
\begin{equation}
V=\{\mbf v\in [H^1(\Omega)]^3\ |\;\;\mbf v=\mbf 0\;\mbox{on } \partial\Omega_u\},
\label{V}
\end{equation}
\begin{equation}
\Sigma=\{\mbf \sigma\in [L^2(\Omega)]^{3\times3}\ |\;\;\sigma_{ij}=\sigma_{ji}\;\mbox{in } \Omega\}.
\label{S}
\end{equation}
Similarly as in LA, the space $V$ represents velocity fields and $\Sigma$ is used for symmetric stress fields. $L^2(\Omega)$ and $H^1(\Omega)$ denotes the Lebesgue and Sobolev spaces, respectively. More advanced functional spaces are considered in \cite{Ch96}.

Using the space $V$ we arrive at the weak form of \eqref{SSR_FoS}$_1$:
\begin{equation}
\int_\Omega\mbf\sigma:\mbf\varepsilon(\mbf v)\,dx=L(\mbf v)\quad\forall \mbf v\in V,
\label{balance}
\end{equation}
where $\mbf\varepsilon$ denotes the strain-rate tensor field,
\begin{equation}
\mbf\varepsilon (\mbf v) = \mbox{$\frac{1}{2}$} (\nabla\mbf v+(\nabla\mbf v)^\top),
\label{strain}
\end{equation}
and $L$ is the load functional defined by 
\begin{equation}
L(\mbf v)=\int_\Omega \mbf F\cdot\mbf v\,dx+\int_{\partial\Omega_f}\mbf f\cdot\mbf v\,ds.
\label{L}
\end{equation}
Let $\Lambda$ denote the set of stresses $\mbf\sigma\in\Sigma$ satisfying \eqref{balance} and let
\begin{equation}
P_{q(\lambda)}:=\{\mbf\sigma\in\Sigma\ |\;\; \Phi\big(q(\lambda);\mbf\sigma\big)\leq0\;\mbox{ in }\Omega\}.
\label{P}
\end{equation}
We see that the set $P_{q(\lambda)}$ represents the constraint \eqref{SSR_FoS}$_2$ and thus we can write
\begin{align}
\omega^*&=\sup\{\lambda\geq0\ |\;\; P_{q(\lambda)}\cap\Lambda\neq\emptyset\}\nonumber\\
&= \sup_{\lambda\geq0}\,\sup_{\sbs\in P_{q(\lambda)}\cap\Lambda}\{\lambda\}.
\label{lambda*}
\end{align}
From \eqref{balance}, we have
\begin{equation*}
\inf_{\sbv\in V}\left[\int_\Omega\mbf\sigma:\mbf\varepsilon(\mbf v)\,dx-L(\mbf v)\right]=\left\{
\begin{array}{cc}
0,& \mbox{if } \mbf\sigma\in \Lambda,\\
-\infty,& \mbox{otherwise}.
\end{array}
\right.
\end{equation*}
Hence, one can rewrite \eqref{lambda*} as follows:
\begin{align}
\omega^*&=\sup_{\lambda\geq0}\,\sup_{\sbs\in P_{q(\lambda)}}\,\inf_{\sbv\in V}\left[\lambda+\int_\Omega\mbf\sigma:\mbf\varepsilon(\mbf v)\,dx-L(\mbf v)\right]\nonumber\\[2mm]
&=\sup_{\lambda\geq0}\,\inf_{\sbv\in V}\,\sup_{\sbs\in P_{q(\lambda)}}\left[\lambda+\int_\Omega\mbf\sigma:\mbf\varepsilon(\mbf v)\,dx-L(\mbf v)\right]\nonumber\\[2mm]
&=\sup_{\lambda\geq0}\,\inf_{\sbv\in V}\left[\lambda+\int_\Omega D(q(\lambda);\mbf\varepsilon(\mbf v))\,dx-L(\mbf v)\right],
\label{lambda*2}
\end{align}
where 
\begin{equation}
D(q(\lambda);\mbf \varepsilon)=\sup_{\substack{\sbs\in\mathbb R^{3\times 3}_{sym}\\ \Phi(q(\lambda);\sbs)\leq0}}\mbf\sigma:\mbf\varepsilon
\label{dissipation}
\end{equation}
denotes the local dissipation function depending on $\lambda\geq0$. The function $D(q(\lambda);\cdot)$ is finite-valued only on a convex cone belonging to $\mathbb R^{3\times 3}_{sym}$. Therefore, the inner problem in \eqref{lambda*2} can be classified as cone programming. \eqref{lambda*2} can be interpreted as the \textit{kinematic principle} of the OPT-SSR method.

Let us note that the ordering of inf and sup has been interchanged during the derivation of \eqref{lambda*2}. The corresponding equality is expected and partially justified by the results presented in \cite{Ch96,HRS19}.  

\section{Regularization method}
\label{sec_regularization}

In \cite{SHHC15,CHKS15,HRS15} and\\ \cite{HRS16,RSH18, HRS19}, a regularization method has been systematically developed for the solution of the limit analysis (LA) problem. This methods has also been used in strain-gradient plasticity \cite{RS20}. The aim of this section is to use the regularization for the solution of the OPT-SSR problem and to study the relation between the original and the regularized problem.

We arise from \eqref{lambda*} and regularize this problem with respect to a parameter $\alpha>0$ as follows:
\begin{equation}
\omega^*_\alpha=\sup_{\lambda\geq0}\,\sup_{\sbs\in P_{q(\lambda)}\cap\Lambda}\left[\lambda-\frac{1}{2\alpha}\int_\Omega\mathbb C^{-1}\mbf\sigma:\mbf\sigma\,dx\right],\label{lambda_alpha1}
\end{equation}
where $\mathbb C$ is a positive definite fourth order tensor, for example, the elastic tensor. We have the following result.
\begin{lemma}
The sequence $\{\omega^*_\alpha\}_{\alpha>0}$ defined by \eqref{lambda_alpha1} is nondecreasing and satisfying
\begin{equation}
\omega^*_\alpha\leq\omega^*,\quad \lim_{\alpha\rightarrow+\infty}\omega^*_\alpha=\omega^*,
\label{omega_alpha}
\end{equation}
where $\omega^*$ is defined by \eqref{lambda*}.
\label{lem_alpha_1}
\end{lemma}
\begin{proof}
It is readily seen that the inequalities $\omega^*_{\alpha_1}\leq\omega^*_{\alpha_2}\leq\omega^*$ hold for any $0 <\alpha_1\leq\alpha_2$. Next, for any $\lambda<\omega^*$ the intersection $P_{q(\lambda)}\cap\Lambda$ is nonempty as follows from \eqref{lambda*}. Then the inner sup-problem in \eqref{lambda_alpha1} has a unique solution  $\mbf\sigma_\lambda\in P_{q(\lambda)}\cap\Lambda$, because it contains the quadratic functional. Consequently, 
$$\lim_{\alpha\rightarrow+\infty}\omega^*_\alpha\geq \lim_{\alpha\rightarrow+\infty}\left[\lambda-\frac{1}{2\alpha}\int_\Omega\mathbb C^{-1}\mbf\sigma_\lambda:\mbf\sigma_\lambda\,dx\right]=\lambda$$
and thus
$$\omega^*\geq\lim_{\alpha\rightarrow+\infty}\omega^*_\alpha\geq \sup\{\lambda\geq0\ |\;\;P_{q(\lambda)}\cap\Lambda\neq\emptyset\}=\omega^*.$$
This implies the limit in \eqref{omega_alpha}.
\end{proof}
One can also write
\begin{equation}
\omega^*_\alpha=\max_{\lambda\geq0}\left[\lambda-G_\alpha(\lambda)\right]=\lambda^*_\alpha-G_\alpha(\lambda^*_\alpha),\label{omega_alpha1}
\end{equation}
where
\begin{align}
G_\alpha(\lambda)&=\inf_{\sbs\in P_{q(\lambda)}\cap\Lambda}\ \frac{1}{2\alpha}\int_\Omega\mathbb C^{-1}\mbf\sigma:\mbf\sigma\,dx\label{G}\\[3mm]
&=\left\{
\begin{array}{cc}
\frac{1}{2\alpha}\int_\Omega\mathbb C^{-1}\mbf\sigma_\lambda:\mbf\sigma_\lambda\,dx, & \mbox{if } P_{q(\lambda)}\cap\Lambda\neq\emptyset,\\
+\infty,& \mbox{otherwise},
\end{array}
\right.\label{G2}
\end{align}
and $\lambda^*_\alpha$ maximizes the middle term in \eqref{omega_alpha1}. Since the value $G_\alpha(\lambda^*_\alpha)$ is finite, we have $P_{q(\lambda^*_\alpha)}\cap\Lambda\neq\emptyset$. This fact, \eqref{lambda*} and \eqref{omega_alpha1} imply the following result.
\begin{lemma}
The sequence $\{\lambda^*_\alpha\}_{\alpha>0}$ defined by \eqref{G} satisfies
\begin{equation}
\omega^*_\alpha\leq\lambda^*_\alpha\leq\omega^*,\quad \lim_{\alpha\rightarrow+\infty}\lambda^*_\alpha=\omega^*,
\label{lambda_alpha}
\end{equation}
where $\omega^*$ and $\omega^*_\alpha$ are defined by \eqref{lambda*} and \eqref{lambda_alpha1}, respectively.
\label{lem_alpha_2}
\end{lemma}
From Lemmas \ref{lem_alpha_1} and \ref{lem_alpha_2}, it follows that the values $\omega^*_\alpha$ and $\lambda^*_\alpha$ are close to $\omega^*$ for sufficiently large $\alpha$. The inequality $\omega^*_\alpha\leq\lambda^*_\alpha\leq\omega^*$  implies that $\lambda_\alpha^*$ is more accurate approximation of $\omega^*$ than $\omega^*_\alpha$. This is illustrated on a numerical example in Section \ref{sec_examples}.

Let us note that the scalar optimization problem in \eqref{omega_alpha1} can be solved, for example, by sequential enlarging $\lambda$ for fixed $\alpha$. A sufficiently large value of $\alpha$ (ensuring that $\lambda_\alpha^*$ is close to $\omega^*$) can be found by a continuation method starting from smaller values of $\alpha$. It suffices to apply the continuation only on a coarse finite element mesh and then a fixed value of $\alpha$ can be used for finer meshes.

Next, for the solution of \eqref{omega_alpha1}, it is crucial to evaluate the function $G_\alpha$. To this end, we use a similar duality approach as presented in Section \ref{sec_duality}. We arrive at the following kinematic definition of $G_\alpha$:
\begin{equation}
G_\alpha(\lambda)=-\inf_{\sbv\in V}\left[\int_\Omega D_\alpha({q(\lambda)};\mbf\varepsilon(\mbf v))\,dx-L(\mbf v)\right],\label{G_alpha}
\end{equation}
where 
\begin{equation}
D_\alpha({q(\lambda)};\mbf \varepsilon)=\sup_{\substack{\sbs\in\mathbb R^{3\times 3}_{sym}\\ \Phi({q(\lambda)};\sbs)\leq0}}\left[\mbf\sigma:\mbf\varepsilon-\frac{1}{2\alpha}\mathbb C^{-1}\mbf\sigma:\mbf\sigma\right]
\label{dissip_alpha}
\end{equation}
is the regularized dissipative function. In particular, $D_\alpha$ is finite-valued and differentiable with respect to $\mbf\varepsilon$ unlike the original dissipation $D$, see, for example, \cite{S14}. Moreover, the second derivative of $D_\alpha$ exists almost everywhere. Let $T_\alpha(q(\lambda);\mbf\varepsilon)\in\mathbb R^{3\times3}_{sym}$ denote the derivative of $D_\alpha(q(\lambda);\mbf\varepsilon)$ with respect to $\mbf\varepsilon$. Then the problem \eqref{G_alpha} is equivalent to the following nonlinear variational equation:
\begin{equation}
\mbox{find } \mbf v_{q(\lambda)}\in V:\quad \int_\Omega T_\alpha(q(\lambda);\mbf\varepsilon(\mbf v_{q(\lambda)})):\mbf\varepsilon(\mbf v)\,dx=L(\mbf v)\quad\forall \mbf v\in V.
\label{v_lambda1}
\end{equation}
It is convenient to solve it by a non-smooth and damped version of the Newton method suggested in \cite{S12}, because this method also finds descent directions of the functional in \eqref{G_alpha}, which do not need to be bounded from below for some $\lambda$. 

Let $D_1:=D_\alpha$ and $T_1:=T_\alpha$ for $\alpha=1$. Then the following formulas hold for any $\alpha>0$, $\lambda\geq0$ and $\mbf\varepsilon\in\mathbb R^{3\times3}_{sym}$:
\begin{equation}
D_\alpha(q(\lambda);\mbf\varepsilon)=\frac{1}{\alpha}D_1(q(\lambda);\alpha\mbf\varepsilon),\quad T_\alpha(q(\lambda);\mbf\varepsilon)=T_1(q(\lambda);\alpha\mbf\varepsilon).
\end{equation}
These formulas simplify the construction of the operators $D_\alpha$ and $T_\alpha$ if the continuation over $\alpha$ is used. In addition, $T_1$ is practically the same as the operator, which arises from the implicit Euler discretization of the elasto-plastic initial-value constitutive problem. Its construction can be found e.g. in \cite{NPO08,SCL17}. The closed form of $D_1$ is presented in the Appendix of this paper.

\section{Numerical examples}
\label{sec_examples}

In this section, we present two different numerical examples on slope stability problems. The first example considers a  homogeneous slope presented in \cite{TSS2015b}. The aim is to illustrate our theoretical results and to verify that the computed FoS are in accordance with the published ones. The second example arises from an analysis of a real slope. Therefore, heterogeneous material conditions and the influence of the pore water pressure are considered in the analysis of this boundary value problem. 

\subsection{Softwares and their numerical solution}

We use and compare the results from three different softwares: in-house codes in Matlab, Plaxis and Comsol Multiphysics.

The in-house Matlab codes are based on elastic-plastic solvers, the finite element method and on mesh adaptivity. They have been systematically developed and described in \cite{SCKKZB16,SCL17,CSV19}. Some of the codes are available for download \cite{CSV18}. Within these codes, we have implemented the regularization method discussed in Section \ref{sec_regularization} to compute safety factors $\lambda^*_{ass}$, $\lambda^*_{A}$, $\lambda^*_{B}$, and $\lambda^*_{C}$ for associated plasticity and for Davis A-C approaches. In particular, six-noded triangular elements with the 7-point Gauss quadrature have been used and combined with the mesh adaptivity introduced in \cite{HRS19,SBKSSP19}. However, 15-noded triangular elements are also implemented in the code.

The software Plaxis enables to solve the shear strength reduction method for both associated and non-associated plasticity. The standard solver is based on the implicit Euler time discretization and the arc-length method \cite{BSE11}. 15-noded triangular elements with a shape function of fourth order are used for the following studies. One can easily implement the Davis A approach in the existing SSR procedure. For the application of the Davis B-C approaches, an iterative procedure is used. Due to the utilization of this software we are able to compare the suggested OPT-SSR method with current approaches of the shear strength reduction (SSR) method.

The software Comsol Multiphysics with its Geomechanical module neither includes the shear strength reduction method nor the arc-length method. However, the code allows to add a global equation to the elastic-plastic system of equations with respect to an unknown parameter and enables an optimization of this parameter. Therefore, the SSR method for associated plasticity and the Davis A-C modifications can be implemented in Comsol Multiphysics. Besides the regularization method, the standard incremental procedure for the solution of the elastic-plastic problem has been used. 15-noded triangular elements are considered.

\subsection{Homogeneous slope} 

Following \cite{TSS2015b}, we consider a homogeneous slope depicted in Figure \ref{fig_scheme}. Its inclination is  $45^\circ$ and sizes (in meters) are given in Figure \ref{fig_scheme}. The effective friction angle $\phi'$ is 45$^\circ$, the effective cohesion $c'$ is 6.0 kPa and the unit weight $\gamma$ is 20.0 kN/m$^3$. The dilatancy angle is either $\psi'=0^\circ$, $\psi'=15^\circ$ or $\psi'=45^\circ$. The chosen values of $\phi'$ and $\psi'$ enable to highlight differences between the associated and non-associated material behavior and between the suggested and the currently used approaches of the SSR method.

Next, we set the following values for the Young modulus and the Poisson ratio: $E=40$ MPa and $\nu=0.3$.
For the regularization method, the value $\alpha=1000$ is used. This value is sufficiently large as will be discussed later.

\begin{figure}[htb]	
\begin{center}
\begin{picture}(220,110)

{\linethickness{1pt}		
	\put(10,10){\line(1,0){200}}
	\put(10,10){\line(0,1){100}}
	\put(210,10){\line(0,1){50}}	
	\put(10,110){\line(1,0){75}}
	\put(135,60){\line(1,0){75}}
	\put(135,60){\line(-1,1){50}}	
}		
\put(8,60){\line(1,0){4}}
\put(85,8){\line(0,1){4}}
\put(135,8){\line(0,1){4}}
\multiput(80,60)(6,0){4}{\vector(0,-1){15}}
\put(88,62){\makebox(0,0)[b]{$\mbf F=-\gamma\mbf e_2$}}

\multiput(13,6)(14,0){15}{\makebox(0,0)[t]{$\triangle$}}
\multiput(6,14)(0,13){8}{\circle{5}}
\multiput(2,11)(0,13){8}{\line(0,1){6}}
\multiput(214,14)(0,13){4}{\circle{5}}
\multiput(218,11)(0,13){4}{\line(0,1){6}}

\put(46,12){\makebox(0,0)[b]{{\scriptsize 15}}}
\put(110,12){\makebox(0,0)[b]{{\scriptsize 10}}}
\put(171,12){\makebox(0,0)[b]{{\scriptsize 15}}}
\put(12,35){\makebox(0,0)[l]{{\scriptsize 10}}}
\put(12,85){\makebox(0,0)[l]{{\scriptsize 10}}}
\put(208,35){\makebox(0,0)[r]{{\scriptsize 10}}}
\end{picture} 
\end{center}
\caption{Geometry of the considered slope with an inclination of $45^\circ$.}
\label{fig_scheme}
\end{figure}
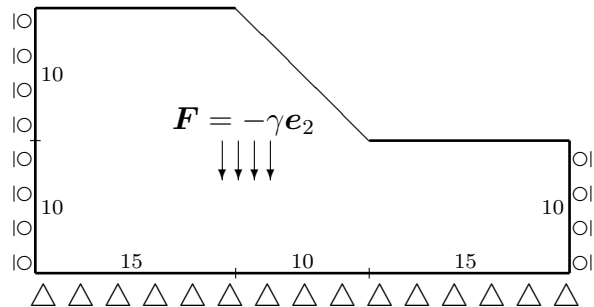

Table \ref{table_FoS1} summarizes the computed factors of safety obtained with different softwares and different approaches. These values are practically the same as discussed in\\ \cite{TSS2015b}. The computed FoS are in accordance with Theorem \ref{theorem_FoS}, that is $1\leq\lambda^*_A\leq \lambda^*_B\leq \lambda^*_C\leq\lambda^*_{assoc}$. We have used the current (standard) approach of the non-associated SSR method only in Plaxis. The corresponding value of FoS for  $\psi=0^\circ$ cannot be uniquely determined due to oscillations of the method, which are a consequence of a varying failure mechanism during the strength reduction procedure (see \cite{TSS2015b} for more details). In \cite{TSS2015b}, the current approach was also investigated for a very fine mesh leading to slightly lower factors of safety, namely 1.42 for $\psi'=15^\circ$ and 1.21–1.27 for $\psi'=0^\circ$. These results indicate again the distinct mesh dependency of FoS in the case of non-associated plasticity.

\begin{table}[htb]
	\caption{Safety factors for the homogeneous slope and different approaches}
	\label{table_FoS1}
	\begin{center}	 
		\begin{tabular}{|l||c|c|c| }
		\hline
		& MATLAB & COMSOL & Plaxis   \\
		\hline\hline
		$\psi=\phi=45^\circ$, assoc. model &  1.52 & 1.52 & 1.51\\
		\hline
		$\psi=15^\circ$, Davis A &  1.27 & 1.28 & 1.27\\
		\hline
		$\psi=15^\circ$, Davis B &  1.36 & 1.37 & 1.35\\
		\hline
		$\psi=15^\circ$, Davis C &  1.41 & 1.42 & 1.41\\
		\hline
		$\psi=15^\circ$, current approach  &   -- & -- & 1.46 \\
		\hline
		$\psi=0^\circ$, Davis A &  1.08 & 1.08 & 1.08\\
		\hline
		$\psi=0^\circ$, Davis B, C &  1.15 & 1.16 & 1.16 \\
		\hline
		$\psi=0^\circ$, current approach  &   -- & -- & 1.27--1.35 \\
		\hline
		\end{tabular}
	\end{center}
\end{table}

Unlike the current SSR approach, the results of OPT-SSR method are practically insensitive if sufficiently fine meshes are used. It is illustrated in Figure \ref{fig_adaptivity1} where the mesh adaptivity within the in-house Matlab codes is used. In this study, 20 levels of meshes are considered and the corresponding safety factors remain from about 1000 elements (level 10) onwards almost constant. The dependence of FoS on the mesh adaptivity has also been analyzed in \cite{OTS18}. The finest mesh and the corresponding slip surface for the Davis B approach are depicted in Figure \ref{fig_finest_mesh1}. The failure surface is visualized using the rate of the deviatoric strain.

\begin{figure}	
	\begin{center}
\includegraphics[width=0.6\textwidth]{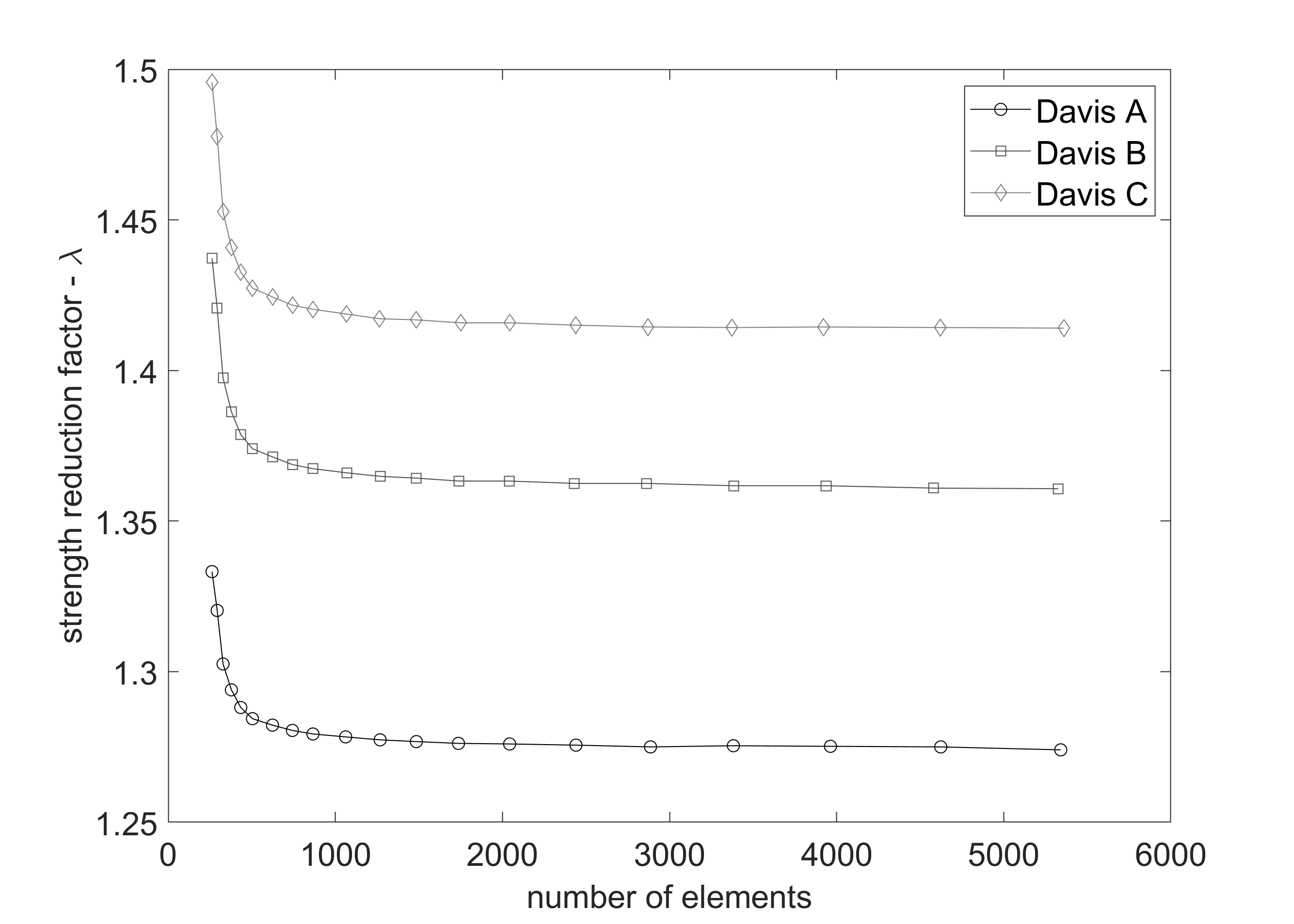}
	\end{center}
	\caption{Safety factors for the homogeneous slope depending on mesh adaptivity ($\psi=15^\circ$).}
	\label{fig_adaptivity1}
\end{figure}

\begin{figure}[h]	
	\begin{center}
		\includegraphics[width=0.45\textwidth]{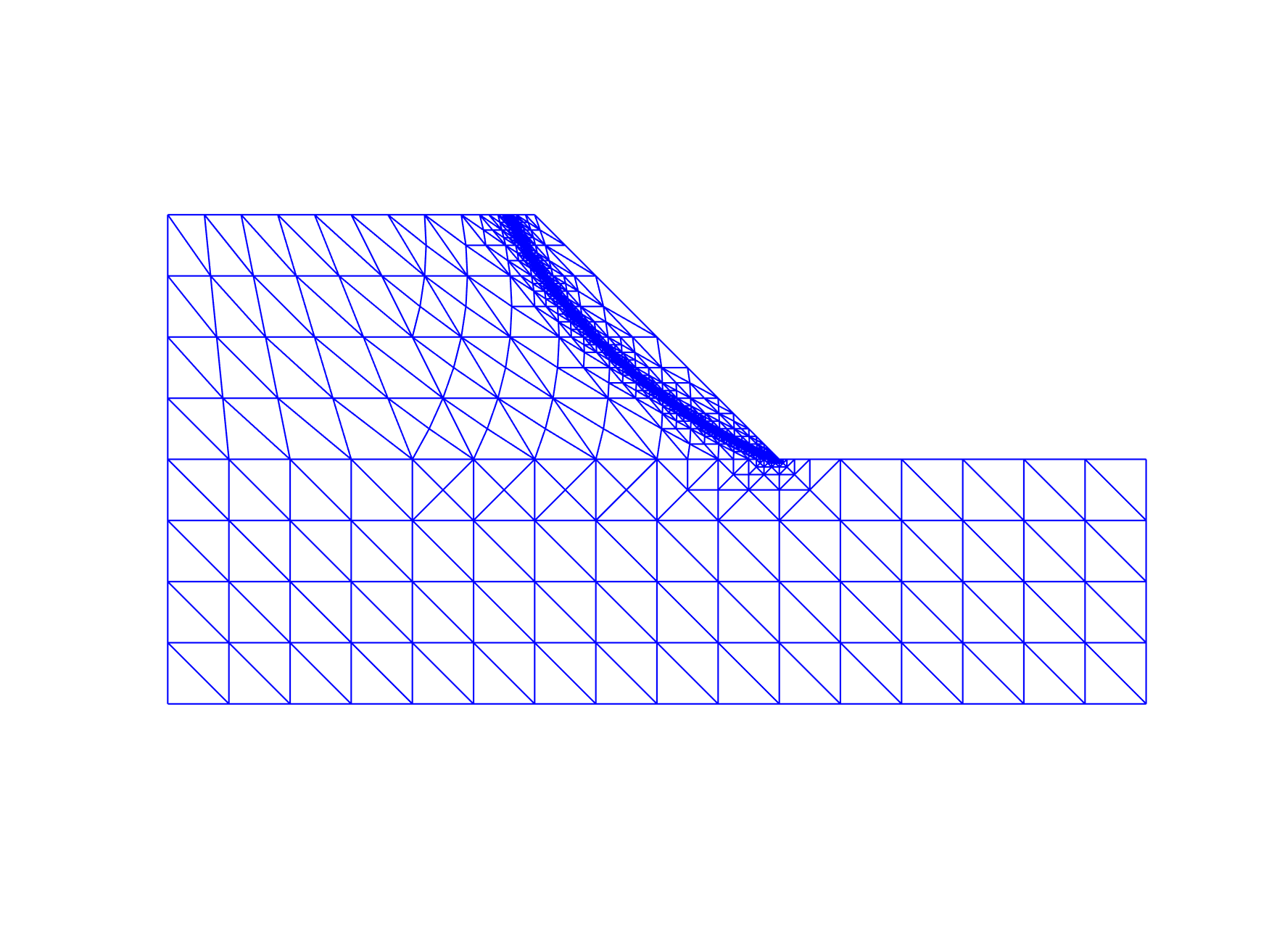}
		\includegraphics[width=0.45\textwidth]{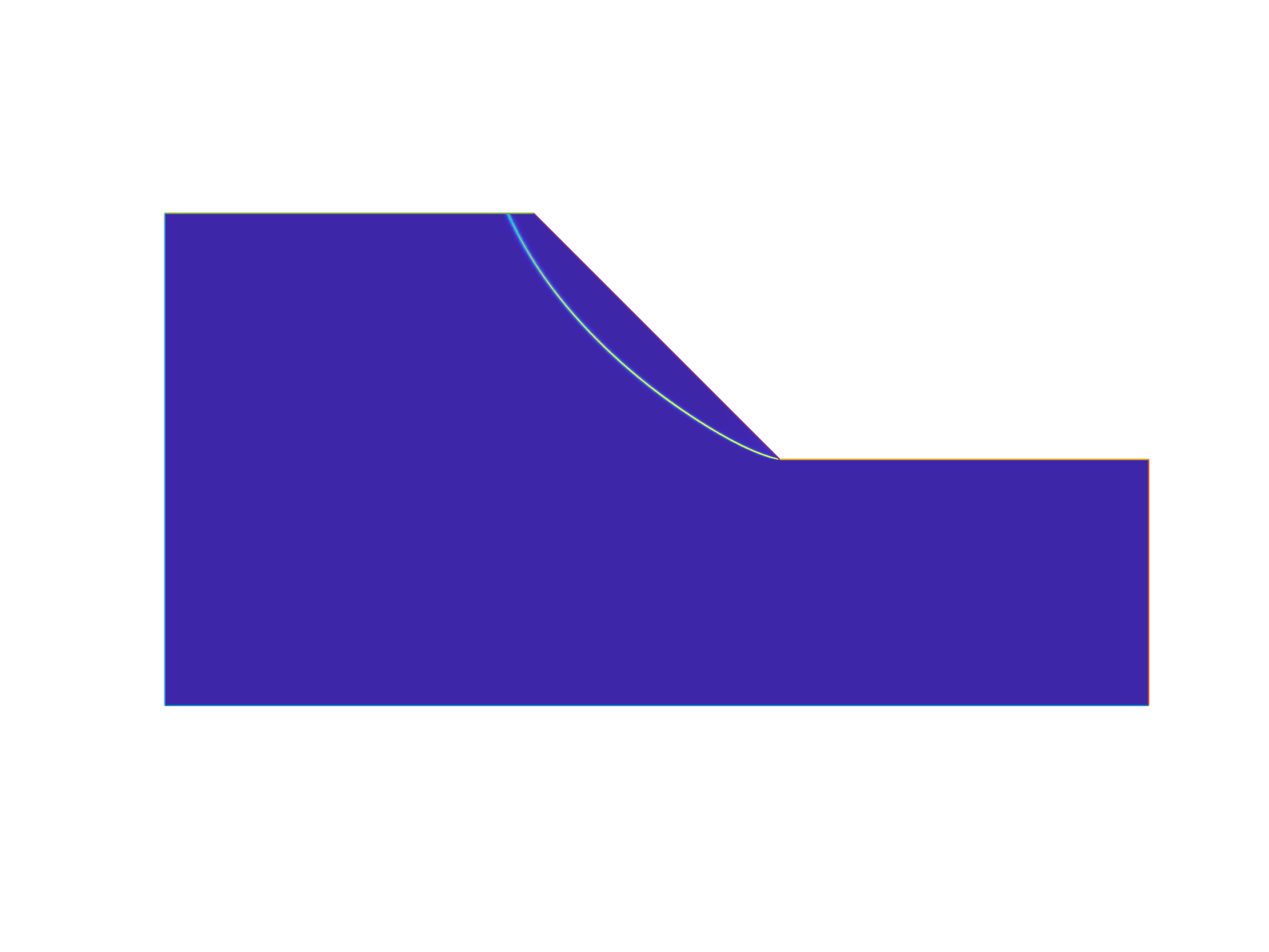}
	\end{center}

\vspace*{-15mm}
	\caption{The finest mesh and the corresponding failure for the homogeneous slope (Matlab code).}
	\label{fig_finest_mesh1}
\end{figure}

The dependency of $\lambda^*_\alpha$ and $\omega^*_\alpha$ (see Section \ref{sec_regularization}) on the regularization parameter $\alpha$ is depicted in Figure \ref{fig_alpha}. One can see that these curves are increasing and approaches $\omega^*$. This confirms the theoretical results of the regularization method. We also see that $\lambda^*_\alpha$ approximates the safety factor $\omega^*$ even for relatively small values of $\alpha$. On the other hand, the bound $\omega^*_\alpha$ is very poor (too far from $\omega^*$) for small values of $\alpha$, see Figure \ref{fig_alpha}(b). Hence, it is important to use sufficiently large values of the regularization parameter $\alpha$.
\begin{figure}[h]	
	\begin{center}
		\includegraphics[width=0.45\textwidth]{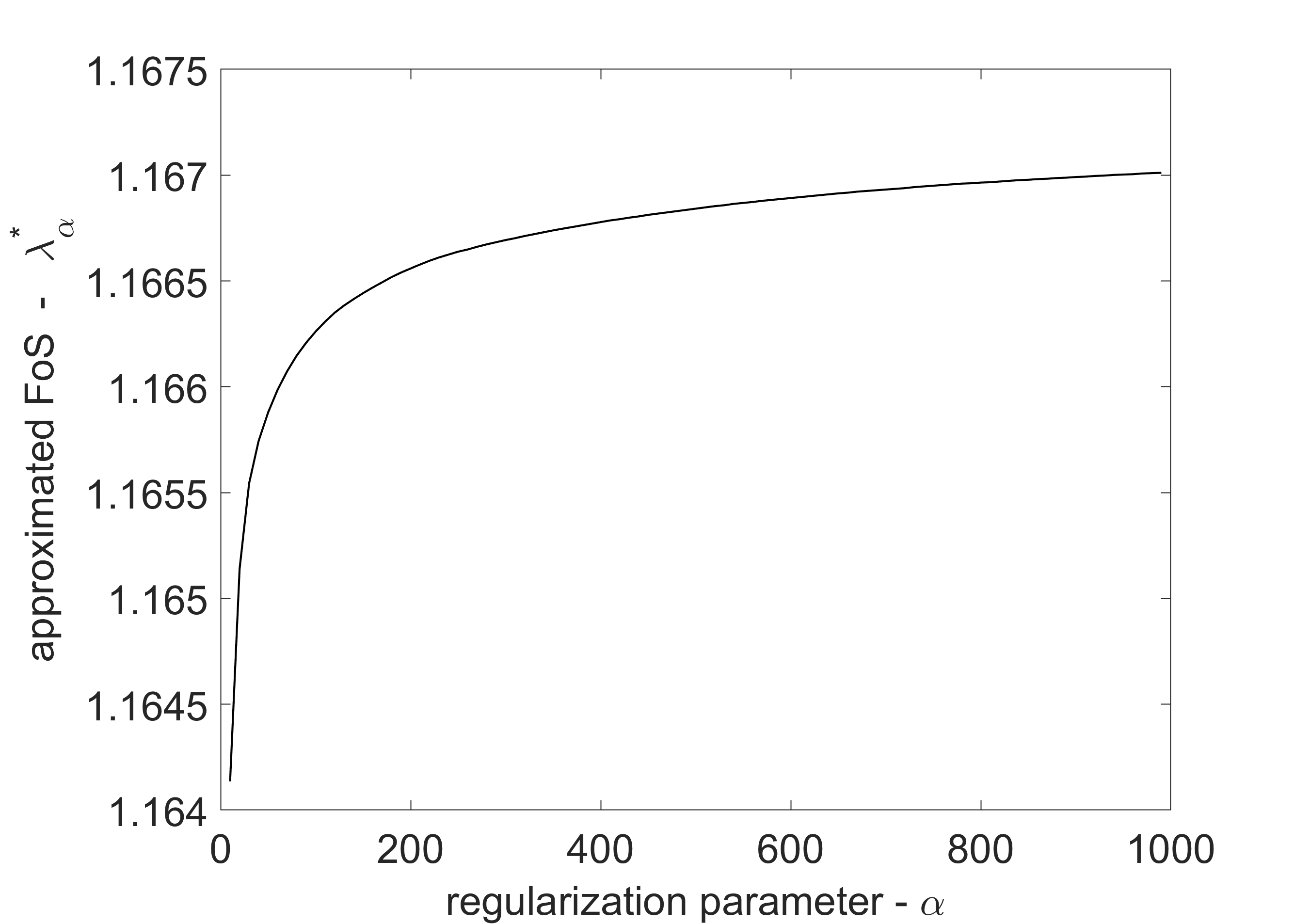}
		\includegraphics[width=0.45\textwidth]{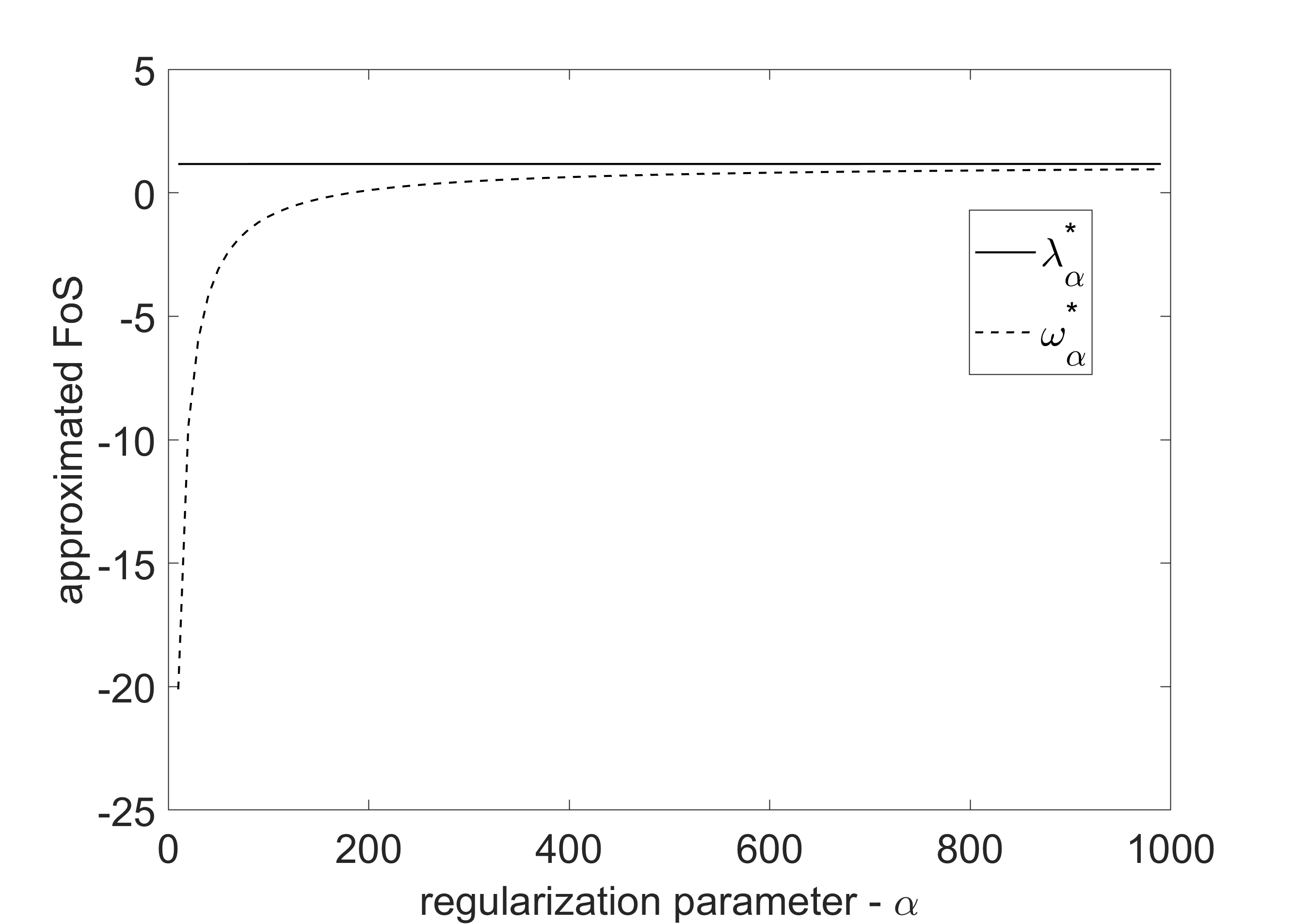}	
		
		\medskip
		(a) \hspace*{7cm} (b)
	\end{center}	
    \vspace{-5mm}
	\caption{(a) Dependence of $\lambda^*_\alpha$ on the regularization parameter $\alpha$. (b) Dependence of $\lambda^*_\alpha$ and $\omega^*_\alpha$ on $\alpha$ -- notice that the scale of the axis $y$ differs from case (a).} 
	\label{fig_alpha}
\end{figure}

\subsection{Case study -- heterogeneous slope from locality Doubrava-Kozinec}

The second example considers a real slope in Doubrava-Kozinec (near Karvina in the North-East part of the Czech Republic). This slope is located within a potentially unstable area with historical manifestations of landslide activity (within the quaternary clay layer). In Figure \ref{fig_scheme2}, the investigated slope including the soil conditions is illustrated. One can see that the slope is heterogeneous and consists of five soil layers. The particular materials and their parameters are specified in Table \ref{table_parameters}. It has to be mentioned that the sand and gravel layers in the investigated slope contain a small amount of silt to clay. For the quaternary clays, we used in our analysis the residual value of the friction angle (due to previous landslide activities). The level of groundwater is indicated by the blue line in Figure \ref{fig_scheme2}. We distinguish the specific weights $\gamma_{\mathrm{unsat}}$ and $\gamma_{\mathrm{sat}}$ for unsaturated and saturated materials, respectively. The values of the dilatancy angle were not available for us (as often the case in practical engineering), therefore, we set $\psi'=0^\circ$ for all materials. However, we also consider the associated case with $\psi'=\phi'$ in order to analyze the influence of the dilatancy angle. Notice that for these choices of $\psi'$, the Davis B and the Davis C approaches coincide.

During the evaluation of the slope, it turned out that the failure mechanism is located in the central part of the slope, or more precisely in quaternary clay and its interface with the clayed sand and the neogene clay layer. It is worth noticing that the effective friction angle of the quaternary clay is much lower than $\phi'$ of the other materials, thus it was expected that this layer is decisive for both, the obtained FoS and the computed failure mechanism. Numerical results presented below confirm the location of the failure mechanism.

\begin{figure}	
\begin{center}
	\includegraphics[width=1\textwidth]{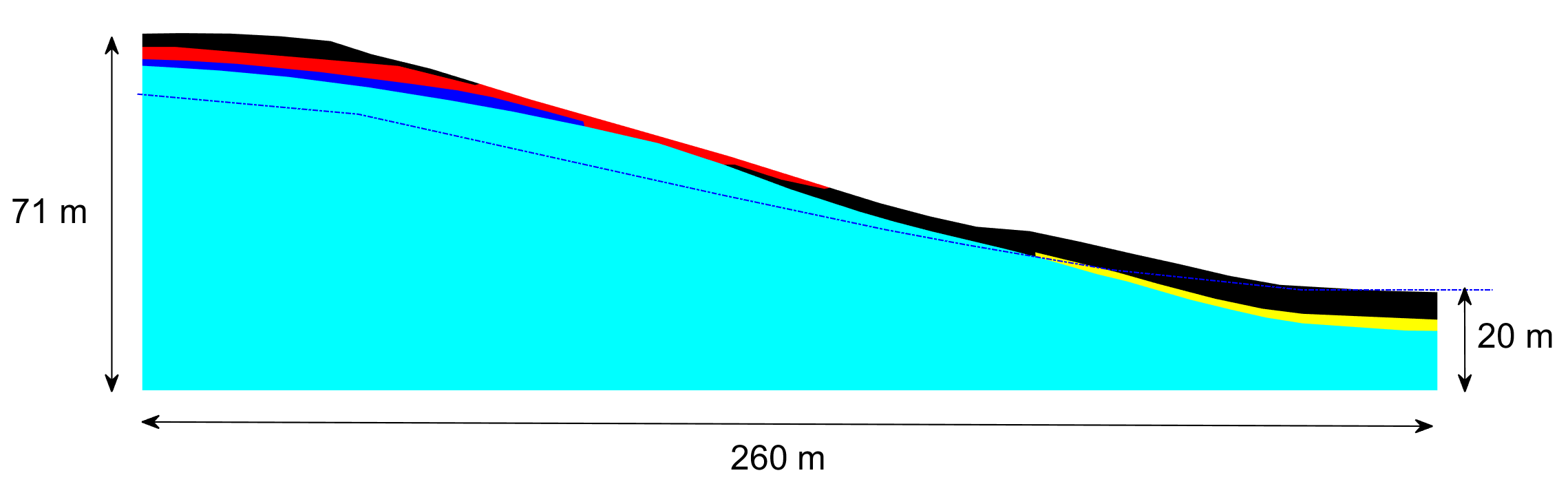}
\end{center}
\vspace*{-5mm}
\caption{Geometry of the case study.}
\label{fig_scheme2}
\end{figure}

\begin{table}[h]
\caption{Material parameters for the heterogeneous slope.}
\label{table_parameters}
\begin{center}	 
	\begin{tabular}{|c||c|c|c|c|c| }
		\hline
		& \crule[cyan]{3mm}{3mm} \scriptsize{neogene clay} & \crule[yellow]{3mm}{3mm} \scriptsize{gravel} & \crule{3mm}{3mm} \scriptsize{quaternary clay} & \crule[blue]{3mm}{3mm} \scriptsize{sand} & \crule[red]{3mm}{3mm} \scriptsize{clayed sand} \\
		\hline\hline
		$\phi'$ [$^\circ$]  &   26 & 45 &  13 & 33 & 27   \\
		\hline
		$c'$ [kPa] &   9 & 1 &  3 & 2 & 5  \\
		\hline
		$E$ [MPa] &   16 & 140 &  10 & 14 & 27  \\
		\hline
		$\nu$  &   0.40 & 0.20 &  0.40 & 0.28 & 0.35  \\
		\hline
		$\gamma_{\mathrm{unsat}}$ [kN/m$^3$] &   20.3 & 20.5 &  20.0 & 19.0 & 19.4  \\
		\hline
		$\gamma_{\mathrm{sat}}$ [kN/m$^3$] &   20.7 & 20.6 &  20.5 & 20.5 & 21.4  \\
		\hline
	\end{tabular}
\end{center}
\end{table}

The initial mesh for the computation in Comsol Multiphysics is depicted in Figure \ref{fig_mesh}. This mesh reflects the  heterogeneity of the soil conditions. This mesh has also been imported to Matlab. In Comsol Multiphysis, this mesh was then locally refined in the central part of the slope to obtain more accurate results in the region of interest (region of the expected failure surface). In Matlab, the original mesh was adaptively refined, where 15 mesh levels were considered. 

\begin{figure}	
	\begin{center}
		\includegraphics[width=0.6\textwidth]{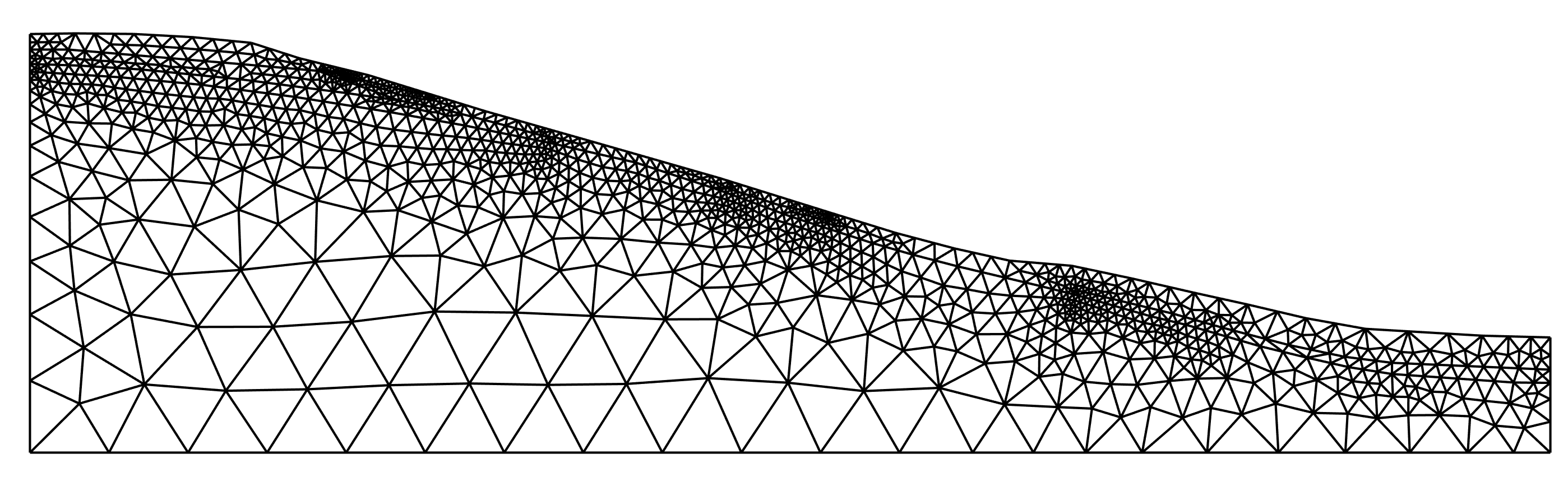}
	\end{center}
\vspace*{-5mm}
	\caption{Initial mesh for the case study (Comsol Multiphysics).}
	\label{fig_mesh}
\end{figure}

A detail of the finest Matlab mesh for $\psi=0^\circ$ and the Davis B approach is depicted in Figure \ref{fig_finest_mesh2} (together with the corresponding failure surface). For the visualization of this zone, a norm of the rate of the deviatoric strain was used. The black curves in the figure depict the soil stratification. One can see that failure mechanism is not very deep and that a large part of the slip surface lies on the transition of the quaternary clay layer to the neogene clay layer. This confirms again that this transition zone is decisive for the stability of the considered slope.

\begin{figure}	
	\begin{center}
		\includegraphics[width=0.4\textwidth]{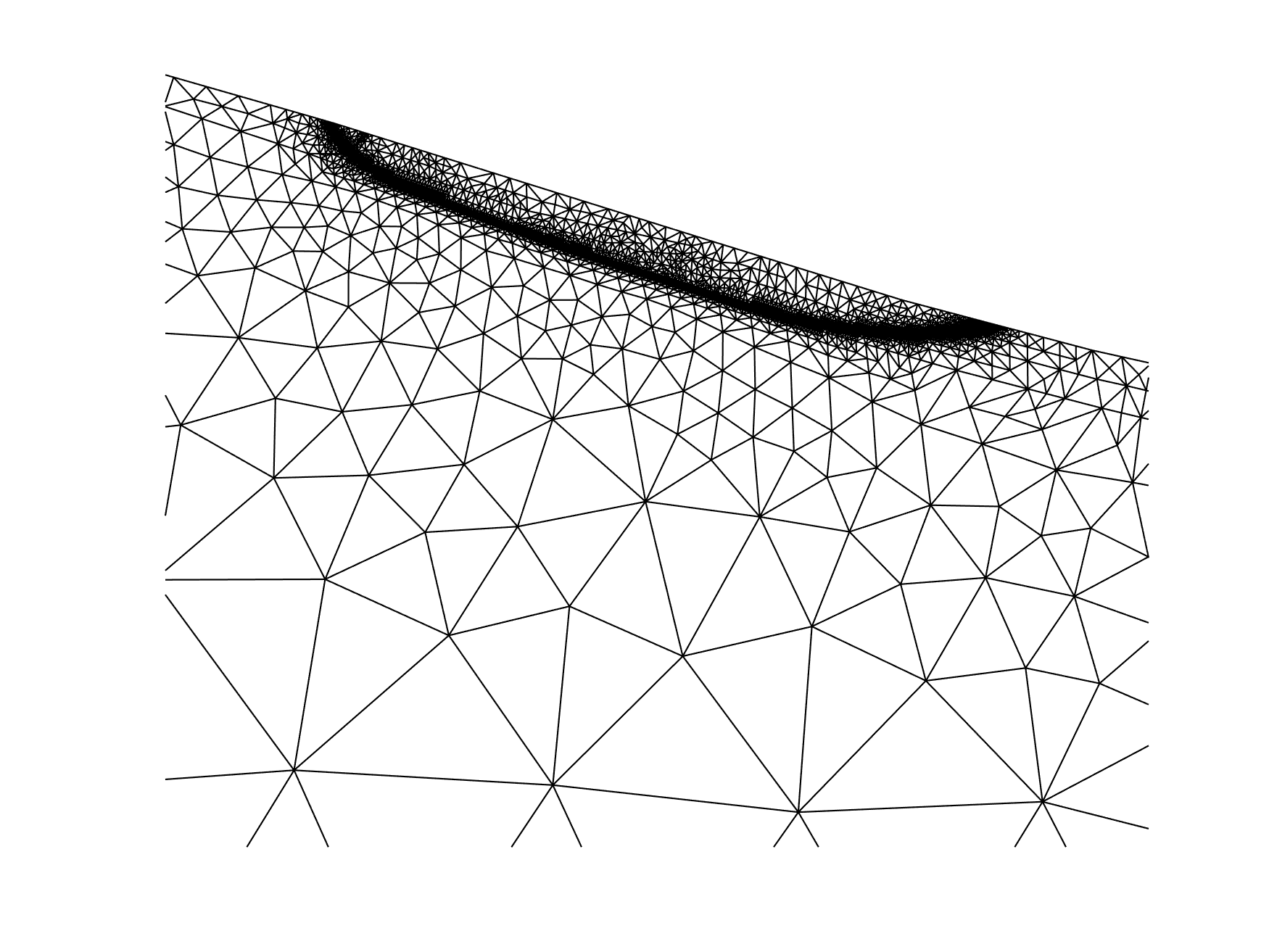}
		\includegraphics[width=0.42\textwidth]{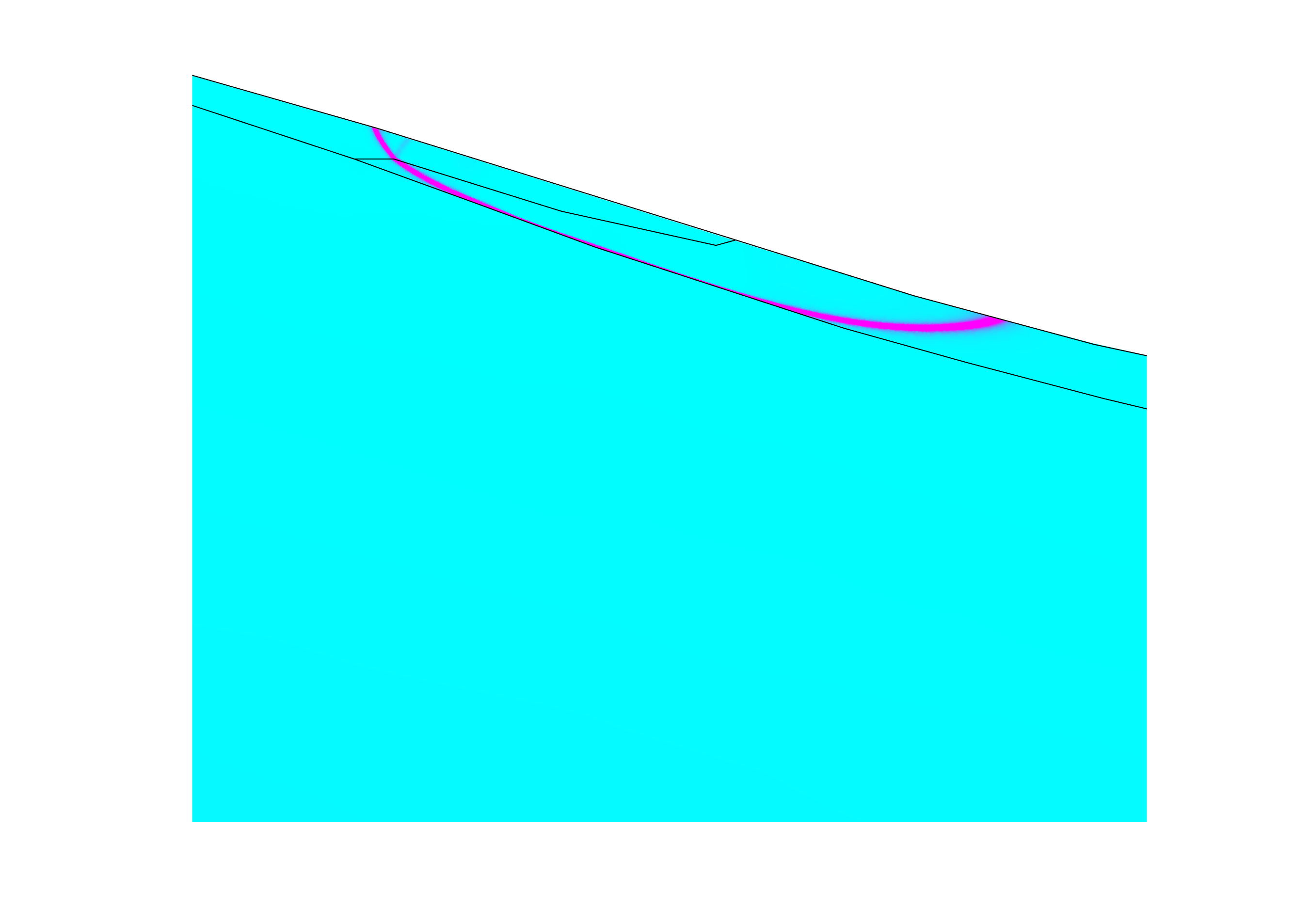}
	\end{center}
\vspace*{-5mm}
	\caption{Detail of the finest mesh used in Matlab (left) and the computed failure surface (right).}
	\label{fig_finest_mesh2}
\end{figure}

In Figure \ref{fig_adaptivity2}, we see the dependence of the safety factors on the mesh adaptivity computed in Matlab (for all three approaches). One can also see that these curves are practically constant after a sufficiently large number of mesh refinements (approximately 10). We also observe that the safety factor obtained with the Davis A approach is slightly lower that the one for Davis B. This is in accordance with Theorem \ref{theorem_FoS}.

\begin{figure}	
	\begin{center}
		\includegraphics[width=0.6\textwidth]{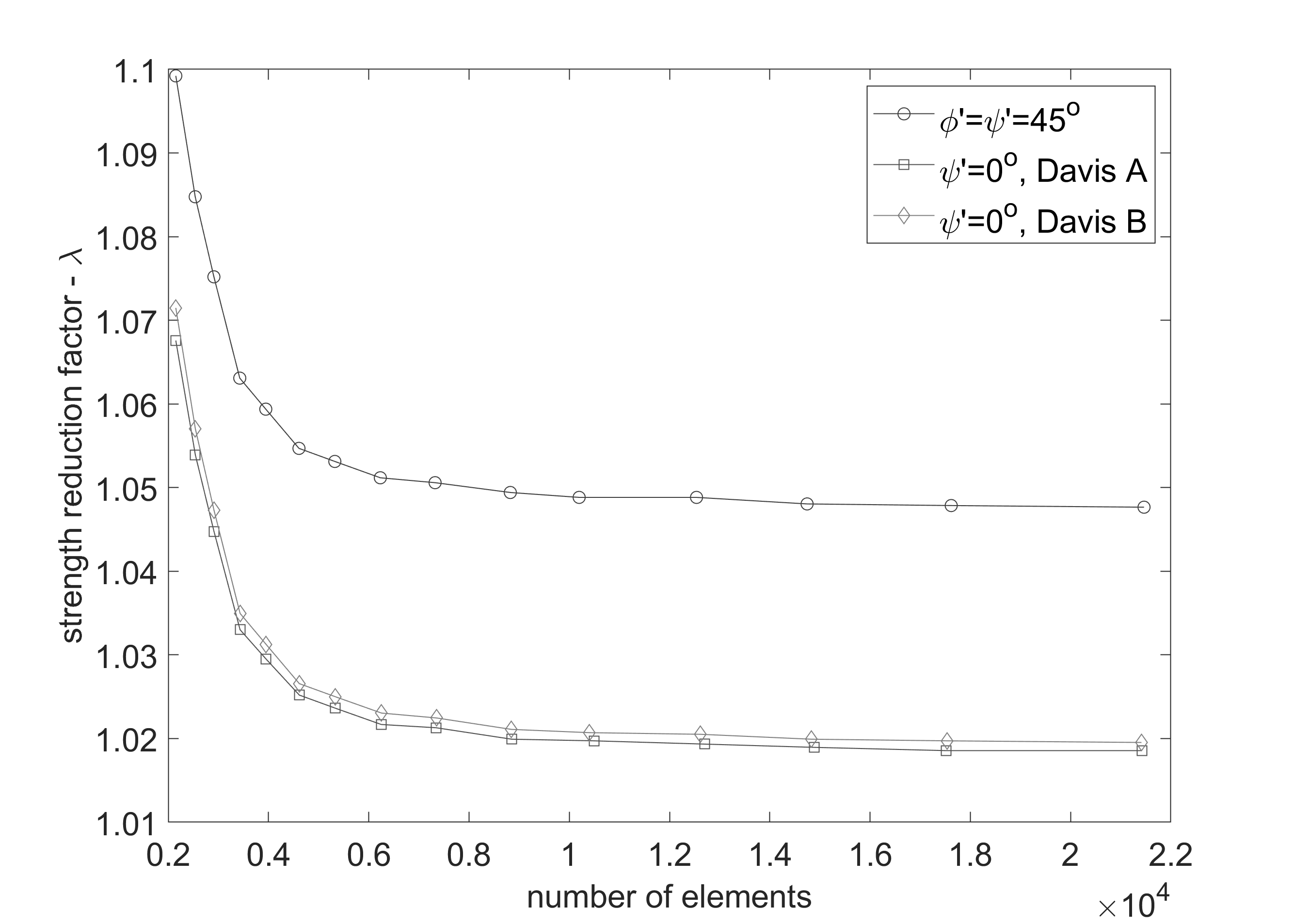}
	\end{center}
\vspace*{-5mm}
	\caption{Safety factors for the case study depending on mesh adaptivity.}
	\label{fig_adaptivity2}
\end{figure}

The computed safety factors for the different approaches are summarized in Table \ref{table_FoS2}. One can see that the safety factors are close to 1 for all investigated approaches indicating that the slope is close to its limit state, thus numerically close to failure. The observed failure of the slope can be explained by inhomogeneity of the quaternary clay layer. It can also be expected that in some parts of the slope the residual friction angle of quaternary clay is slightly lower than 13$^\circ$, which in turn causes a reduction of the factor of safety. The values computed with the in-house Matlab codes are slightly lower than other FoS values, that is due to the usage of the local mesh adaptivity. We also observe that FoS for the Davis A and Davis B approaches practically coincide and are very close to one. This observation is in accordance with the third statement of Theorem \ref{theorem_FoS}.

\begin{table}[htb]
	\caption{Safety factors for the case study using different approaches}
	\label{table_FoS2}
	\begin{center}	 
		\begin{tabular}{|c||c|c|c| }
			\hline
			& MATLAB & COMSOL & Plaxis  \\
			\hline\hline
			$\psi'=\phi'$, assoc. model &  1.05 & 1.09 & 1.08\\
			\hline
			$\psi'=0^\circ$, Davis A &  1.02 & 1.06 & 1.05\\
			\hline
			$\psi'=0^\circ$, Davis B & 1.02 & 1.06 & 1.06 \\
			\hline
			$\psi'=0^\circ$, current approach  &   -- & -- & 1.06 \\
			\hline
		\end{tabular}
	\end{center}
\end{table}

\section{Conclusion}
\label{sec_conclusion}

This work has been inspired by the recent paper \cite{TSS2015b}, where the standard shear strength reduction (SSR) method was approximated by modified Davis approaches and a parametrized limit analysis (LA) method. Based on these ideas, the paper presents an abstract optimization problem (OPT-SSR) related to the SSR method. Next, it is shown that the Davis A-C modifications suggested in \cite{TSS2015b} can be written in the form of the OPT-SSR problem. This fact has been verified on numerical investigations. 
The suggested approach has been completed by variational principles and duality theory, similarly as in limit analysis. Hence, the OPT-SSR approach can be interpreted as a rigorous method.

For the numerical solution, a regularization method has been used and combined with the finite element and the damped Newton method. This solution concept leads to similar solvers as standardly used in computational plasticity and thus can be easily implemented within existing elastic-plastic codes. In particular, in-house Matlab codes \cite{CSV19} in combination with local mesh adaptivity have been used. Softwares Plaxis and Comsol Multiphysics have been utilized for comparison of the results. One of the presented numerical examples can be classified as case history, since it deals with
a real slope.

{\bf Acknowledgment:} The authors acknowledge support for their work from the Czech Science Foundation (GA\v{C}R) through project No. 19-11441S. The authors also thank to Dr. Alexej Kolcun for fruitful discussions on mesh adaptivity for regular and irregular meshes.

\section{Appendix -- closed form of the function $D_1$}

For the sake of completeness, we introduce the closed form of the function
\begin{equation}
D_1(q(\lambda);\mbf \varepsilon)=\sup_{\substack{\sbs\in\mathbb R^{3\times 3}_{sym}\\ \Phi(q(\lambda);\sbs)\leq0}}\left[\mbf\sigma:\mbf\varepsilon-\frac{1}{2}\mathbb C^{-1}\mbf\sigma:\mbf\sigma\right].
\label{dissip_alpha1}
\end{equation}
from Section \ref{sec_regularization}. In literature (see, for example, \cite{NPO08} and\\ \cite{SCL17}), one can find closed form of the derivative $T_1$ of $D_1$ representing the stress-strain relation but not $D_1$. Therefore, we try to fill this gap.

Beside $\lambda\geq0$, this function also depends on the parameters $c'$, $\phi'$ and the elastic parameters $K$, $G$. We construct the function $D_1(q(\lambda);\mbf \varepsilon)$ only for such $\lambda$ satisfying $q(\lambda)=1$. For other choices of $\lambda$, it suffices to replace $c'$ and $\phi'$ with 
\begin{equation}
c_{q(\lambda)}:=\frac{c'}{q(\lambda)},\quad  \phi_{q(\lambda)}:=\arctan\frac{\tan\phi'}{q(\lambda)}.
\label{reduction2}
\end{equation}
see also \eqref{reduction}. To be in accordance with the derivation presented in \cite{SCL17}, we write the yield criterion $\Phi(1;\mbf\sigma)\leq0$ in the form
\begin{equation}
(1+\sin\phi')\sigma_1-(1-\sin\phi')\sigma_3-2 c'\,\cos\phi'\leq 0.
\label{MC3}
\end{equation}
Next, we assume that the strain tensor $\mbf\varepsilon$ is given and its eigenvalues $\varepsilon_1$, $\varepsilon_2$, $\varepsilon_3$ satisfy $\varepsilon_1\geq\varepsilon_2\geq\varepsilon_3$. Let $\mathrm{tr}\,\mbf\varepsilon=\varepsilon_1+\varepsilon_2+\varepsilon_3$ denote the trace of $\mbf\varepsilon$ and $\Lambda=\frac{1}{3}(3K-2G)$ denote the first Lam\'e coefficient. As in \cite{NPO08}, we distinguish five possible cases: the elastic response, the return to the smooth portion of the Mohr-Coulomb pyramid, the return to the left edge, the return to the right edge, and the return to the apex of the pyramid. 

\paragraph{The elastic response.} This case happen if the elastic stress $\mathbb C\mbf\varepsilon$ satisfies $\Phi(1;\mathbb C\mbf\varepsilon)\leq0$, that is,
\begin{equation}
2\Lambda(\mathrm{tr}\,\mbf\varepsilon)\sin\phi'+2G(1+\sin\phi')\varepsilon_1-2G(1-\sin\phi')\varepsilon_3-2 c\,\cos\phi'\leq 0.
\label{elast_case}
\end{equation}
Then 
\begin{equation}
D_1(1;\mbf \varepsilon)=\frac{1}{2}\mathbb C\mbf\varepsilon:\mbf\varepsilon=\frac{1}{2}\Lambda(\mathrm{tr}\,\mbf\varepsilon)^2+G(\varepsilon_1^2+\varepsilon_2^2+\varepsilon_3^2).
\end{equation}
If the criterion \eqref{elast_case} does not hold then the plastic response occurs and we distinguish four possible cases of the return to the Mohr-Coulomb pyramid. We use the following auxiliary notation:
$$\gamma_{s,l}=\frac{\varepsilon_1-\varepsilon_2}{1+\sin\phi'},\quad \gamma_{s,r}=\frac{\varepsilon_2-\varepsilon_3}{1-\sin\phi'},$$
$$\gamma_{l,a}=\frac{\varepsilon_1+\varepsilon_2-2\varepsilon_3}{3-\sin\phi'},\quad \gamma_{r,a}=\frac{2\varepsilon_2-\varepsilon_2-\varepsilon_3}{3+\sin\phi'}.$$

\paragraph{The return to the smooth portion.} This case happen if \eqref{elast_case} is not satisfied and 
\begin{equation}
q_s(\mbf\varepsilon)<S\min\{\gamma_{s,l},\gamma_{s,r}\},
\end{equation}
where 
$$q_s(\mbf\varepsilon)=2\Lambda(\mathrm{tr}\,\mbf\varepsilon)\sin\phi'+2G(1+\sin\phi')\varepsilon_1-2G(1-\sin\phi')\varepsilon_3-2 c'\,\cos\phi',$$
$$S=4\Lambda\sin^2\phi'+4G(1+\sin^2\phi').$$
Then,
\begin{equation}
D_1(1;\mbf \varepsilon)=\frac{1}{2}\Lambda(\mathrm{tr}\,\mbf\varepsilon)^2+G(\varepsilon_1^2+\varepsilon_2^2+\varepsilon_3^2)-\frac{1}{2S}q_s^2(\mbf\varepsilon).
\end{equation}

\paragraph{The return to the left edge.} This case happen if \eqref{elast_case} is not satisfied and 
\begin{equation}
\gamma_{s,l}<\gamma_{l,a},\quad L\gamma_{s,l}\leq q_l(\mbf\varepsilon)< L\gamma_{l,a},
\end{equation}
where 
$$q_l(\mbf\varepsilon)=2\Lambda(\mathrm{tr}\,\mbf\varepsilon)\sin\phi'+G(1+\sin\phi')(\varepsilon_1+\varepsilon_2)-2G(1-\sin\phi')\varepsilon_3-2 c'\,\cos\phi',$$
$$L=4\Lambda\sin^2\phi'+G(1+\sin\phi')^2+2G(1-\sin\phi')^2.$$
Then,
\begin{equation}
D_1(1;\mbf \varepsilon)=\frac{1}{2}\Lambda(\mathrm{tr}\,\mbf\varepsilon)^2+G\Big[\frac{1}{2}(\varepsilon_1+\varepsilon_2)^2+\varepsilon_3^2\Big]-\frac{1}{2L}q_l^2(\mbf\varepsilon).
\end{equation}

\paragraph{The return to the right edge.} This case happen if \eqref{elast_case} is not satisfied and 
\begin{equation}
\gamma_{s,r}<\gamma_{r,a},\quad R\gamma_{s,r}\leq q_r(\mbf\varepsilon)< R\gamma_{r,a},
\end{equation}
where 
$$q_r(\mbf\varepsilon)=2\Lambda(\mathrm{tr}\,\mbf\varepsilon)\sin\phi'+2G(1+\sin\phi')\varepsilon_1-G(1-\sin\phi')(\varepsilon_2+\varepsilon_3)-2 c'\,\cos\phi',$$
$$R=4\Lambda\sin^2\phi'+2G(1+\sin\phi')^2+G(1-\sin\phi')^2.$$
Then,
\begin{equation}
D_1(1;\mbf \varepsilon)=\frac{1}{2}\Lambda(\mathrm{tr}\,\mbf\varepsilon)^2+G\Big[\varepsilon_1^2+\frac{1}{2}(\varepsilon_2+\varepsilon_3)^2\Big]-\frac{1}{2R}q_r^2(\mbf\varepsilon).
\end{equation}

\paragraph{The return to the apex.} This case happen if \eqref{elast_case} is not satisfied and 
\begin{equation}
q_a(\mbf\varepsilon)\geq A\max\{\gamma_{l,a},\gamma_{r,a}\},
\end{equation}
where 
$$q_a(\mbf\varepsilon)=2K(\mathrm{tr}\,\mbf\varepsilon)\sin\phi'-2 c'\,\cos\phi',\quad A=4K\sin^2\phi'.$$
Then,
\begin{equation}
D_1(1;\mbf \varepsilon)=\frac{1}{2}K(\mathrm{tr}\,\mbf\varepsilon)^2-\frac{1}{2A}q_a^2(\mbf\varepsilon)=\frac{c'}{\tan\phi'}(\mathrm{tr}\,\mbf\varepsilon)-\frac{(c')^2}{2K\tan^2\phi'}.
\end{equation}

\section*{Nomenclatures}

The notation below is chronologically ordered.

\begin{tabular}{ll}
SSR & Shear strength reduction\\
OPT-SSR & Optimization variant of the strength reduction method\\
LA & Limit analysis\\
FoS & Factor of safety\\
$c'$ & Effective cohesion\\
$\phi'$ & Effective friction angle\\
$\psi'$ & Dilatancy angle\\
$\lambda$ & Control parameter for strength reduction, $\lambda\geq0$ \\
$c_\lambda$ & Reduced cohesion depending on $\lambda$\\
$\phi_\lambda$ & Reduced friction angle depending on $\lambda$\\
$\psi_\lambda$ & Reduced dilatancy angle depending on $\lambda$\\
$\mbf \sigma$ & Cauchy stress tensor\\
$\sigma_1, \sigma_3$ & Maximal and minimal principal stresses of $\mbf\sigma$\\
$\Phi$ & Mohr-Coulomb yield function \\
$\Omega$ & Bounded domain represented an investigated body\\
$\partial\Omega_f$ & A part of the boundary of $\Omega$\\
$\mbf n$ & The outward unit normal to the boundary of $\Omega$\\
$\mbf F, \mbf f$ & Volume and surface forces\\
$q(\lambda)$ & Function defining the strength reduction. We distinguish the following\\
& particular choices of $q$: $q_{ass}$, $q_{A}$, $q_{B}$, and $q_{C}$ for the associated model and\\
& for Davis A-C approaches.\\
$\omega^*$ & Factor os safety for the OPT-SSR method. We distinguish the following\\
& particular choices of $\omega^*$: $\lambda^*_{ass}$, $\lambda^*_{A}$, $\lambda^*_{B}$, and $\lambda^*_{C}$ for the associated model\\
& and for Davis A-C approaches.
\end{tabular}

\begin{tabular}{ll}
$\zeta^*$ & Load factor for the limit analysis method\\
$\ell(\lambda)$ & Factor of safety for limit analysis depending on the strength parameter $\lambda$\\
$V$ & Space of velocity field\\
$\Sigma$ & Space of stress fields\\
$\mbf v$ & Velocity field\\
$\mbf \varepsilon(\mbf v)$ & Strain-rate tensor depending on $\mbf v$\\
$L$ & Functional of external forces\\
$P_{q(\lambda)}$ & Set of plastically admissible stress fields\\
$\Lambda$ & Set of statically admissible stress fields\\
$D$ & Local dissipation function\\
$\alpha$ & Regularization parameter\\
$\omega^*_\alpha$, $\lambda^*_\alpha$ & Approximations of $\omega$ given by the regularization\\
$D_\alpha$ & Regularized local dissipation function, $D_1=D_\alpha$ for $\alpha=1$\\
$T_\alpha$ & Derivative of $D_\alpha$ with respect to the stress variable, $T_1=T_\alpha$ for $\alpha=1$\\
$\gamma$ & Unit weight for a homogeneous slope\\
$E$ & Young's modulus\\
$\nu$ & Poisson's ratio\\
$\gamma_{\mathrm{unsat}}$ & Specific weight for unsaturated material\\
$\gamma_{\mathrm{sat}}$ & Specific weight for saturated material\\

\end{tabular}



\begin{thebibliography}{00}
	
{\small
	

\bibitem[Brinkgreve and Bakker (1991)]{BB91}
Brinkgreve, R. B. J. \& Bakker, H. L. (1991). Non-linear finite
element analysis of safety factors. Proceedings of the international
conference on computer methods and advances in geomechanics,
pp. 1117--1122. Rotterdam, the Netherlands: Balkema.

\bibitem[Brinkgreve (2011)]{BSE11}
Brinkgreve, R. B. J., Swolfs, W. M. \& Engin, E. (2011). Plaxis 2D
2011 -- user manual. Delft, the Netherlands: Plaxis bv.
	
\bibitem[Cermak et al. (2015)]{CHKS15}
Cermak, M., Haslinger, J., Kozubek, T., Sysala, S. (2015). Discretization and numerical realization of contact problems for elastic‐perfectly plastic bodies. PART II–numerical realization, limit analysis. {\em ZAMM‐Journal of Applied Mathematics and Mechanics/Zeitschrift für Angewandte Mathematik und Mechanik}, 95(12), 1348--1371.

\bibitem[Cermak et al. (2018)]{CSV18}
\v Cerm\'ak, M., Sysala, S., \& Valdman, J. (2018). MATLAB FEM package for elastoplasticity,  https://github.com/matlabfem/matlab\_fem\_elastoplasticity.

\bibitem[Cermak et al. (2019)]{CSV19}
\v Cerm\'ak, M., Sysala, S., \& Valdman, J. (2019). Efficient and flexible MATLAB implementation of 2D and 3D elastoplastic problems. Applied Mathematics and Computation, 355, 595--614.

\bibitem[Chen and Liu (1990)]{CL90}
Chen, W. and Liu, X.L. (1990). {\em Limit Analysis in Soil Mechanics}.  Elsevier.

\bibitem[Christiansen (1996)]{Ch96}
Christiansen, E. (1996). {\em Limit analysis of colapse states}. In P.~G. Ciarlet and J.~L. Lions, editors, {\em Handbook of Numerical Analysis}, Vol IV, Part 2, North-Holland, 195--312.

\bibitem[Davis (1968)]{D68}
Davis, E. H. (1968). Theories of plasticity and failure of soil
masses. In Soil mechanics: selected topics (ed. I. K. Lee), pp.
341--354. New York, NY, USA: Elsevier.

\bibitem[Dawson et al. (1999)]{DRD99}
Dawson, E. M., Roth, W. H. \& Drescher, A. A. (1999). Slope
stability analysis by strength reduction. G\'eotechnique 49, No. 6,
835--840, http://dx.doi.org/10.1680/geot.1999.49.6.835.

\bibitem[Duncan (1996)]{D96}
Duncan, J. M. (1996). State of the art: limit equilibrium and finite-element analysis of slopes. Journal of Geotechnical engineering, 122(7), 577--596.

\bibitem[de Souza Neto et al. (2011)]{NPO08}
de Souza Neto, E. A., Peric, D., \& Owen, D. R. (2011). Computational methods for plasticity: theory and applications. John Wiley \& Sons.

\bibitem[Ekeland and Temam (1974)]{ET74}
Ekeland, I. and Temam, R. (1974).
\textit{Analyse Convexe et Probl\`emes Variationnels}. 
Dunod, Gauthier Villars, Paris.

\bibitem[Griffiths and Lane (1999)]{GL99}
Griffiths, D. V. \& Lane, P. A. (1999). Slope stability analysis by
finite elements. G\'eotechnique 49, No. 3, 387--403, http://
dx.doi.org/10.1680/geot.1999.49.3.387.

\bibitem[Hamlaoui et al. (2017)]{HQS17}
Hamlaoui, M., Oueslati, A., \& De Saxcé, G. (2017). A bipotential approach for plastic limit loads of strip footings with non-associated materials. International Journal of Non-Linear Mechanics, 90, 1--10.

\bibitem[Han and Reddy (2012)]{HR12}
Han, W., \& Reddy, B. D. (2012). Plasticity: mathematical theory and numerical analysis (Vol. 9). Springer Science \& Business Media.

\bibitem[Haslinger et al. (2016a)]{HRS15}
Haslinger, J., Repin, S., Sysala, S. (2016). A reliable incremental method of computing the limit load in deformation plasticity based on compliance: Continuous and discrete setting. {\em Journal of Computational and Applied Mathematics}  {303}, 156--170.
	
\bibitem[Haslinger et al. (2016b)]{HRS16}
Haslinger, J., Repin, S., Sysala, S (2016). Guaranteed and computable bounds of the limit load for variational problems with linear growth energy functionals. \textit{Applications of Mathematics} {61}, 527--564.

\bibitem[Haslinger et al. (2019)]{HRS19}
Haslinger, J., Repin, S., Sysala, S. (2019). Inf-sup conditions on convex cones and applications to limit load analysis. \textit{Mathematics and Mechanics of Solids} {24}, 3331--3353.

\bibitem[Hjiaj et al. (2003)]{HFS03}
Hjiaj, M., Fortin, J., \& de Saxcé, G. (2003). A complete stress update algorithm for the non-associated Drucker–Prager model including treatment of the apex. International Journal of Engineering Science, 41(10), 1109--1143.


\bibitem[Krabbenhoft et al. (2012)]{KKLS2012}
Krabbenhoft, K., Karim, M. R., Lyamin, A. V., \& Sloan, S. W. (2012). Associated computational plasticity schemes for nonassociated frictional materials. International Journal for Numerical Methods in Engineering, 90(9), 1089--1117.

\bibitem[Krabbenhoft et al. (2016)]{KLK16}
Krabbenhoft, K., Lyamin, A., \& Krabbenhoft, J. (2016). OptumG2: theory. Newcastle, Australia: Optum Computational Engineering.

\bibitem[Michalowski and Drescher (2009)]{MD09}
Michalowski, R.L., \& Drescher, A. (2009). Three-dimensional stability of slopes and excavations. Géotechnique, 59(10), 839--850.



\bibitem[Oberhollenzer et al. (2018)]{OTS18}
Oberhollenzer, S., Tschuchnigg, F., \& Schweiger, H. F. (2018). Finite element analyses of slope stability problems using non-associated plasticity. Journal of Rock Mechanics and Geotechnical Engineering, 10(6), 1091-1101.


\bibitem[Reddy and Sysala (2020)]{RS20}
Reddy, B.D. \& Sysala, S. (2020). Bounds on the elastic threshold for problems of dissipative strain-gradient plasticity. Journal of the Mechanics and Physics of Solids 143, 104089.

\bibitem[Repin et al. (2018)]{RSH18}
Repin, S., Sysala, S., Haslinger, J. Computable majorants of the limit load in Hencky's plasticity problems. \textit{Comp. \& Math. with Appl.} (2018) {75}: 199--217.

\bibitem[Schofield (2005)]{Sch05}
Schofield, A. N. (2005). Disturbed soil properties and geotechnical design. Thomas Telford.

\bibitem[Sloan (2013)]{Sl13}
Sloan SW (2013). Geotechnical stability analysis, {\em G\'eotechnique}, {63}, 531--572.

\bibitem[Sysala (2012)]{S12}
Sysala, S. (2012). Application of a modified semismooth Newton method to some elasto-plastic problems. Mathematics and Computers in Simulation, 82(10), 2004--2021.

\bibitem[Sysala (2014)]{S14}
Sysala, S. (2014). Properties and simplifications of constitutive time‐discretized elastoplastic operators. ZAMM‐Journal of Applied Mathematics and Mechanics/Zeitschrift für Angewandte Mathematik und Mechanik, 94(3), 233--255.

\bibitem[Sysala et al. (2019)]{SBKSSP19}
Sysala, S., Blaheta, R., Kolcun, A., Ščučka, J., Souček, K., \& Pan, P. Z. (2019). Computation of Composite Strengths by Limit Analysis. Key Engineering Materials, 810, 137--142. 

\bibitem[Sysala et al. (2017)]{SCL17}
Sysala, S., Čermák, M., \& Ligurský, T. (2017). Subdifferential‐based implicit return‐mapping operators in Mohr‐Coulomb plasticity. ZAMM‐Journal of Applied Mathematics and Mechanics/Zeitschrift für Angewandte Mathematik und Mechanik, 97(12), 1502--1523.

\bibitem[Sysala et al. (2015)]{SHHC15}
Sysala, S., Haslinger, J., Hlaváček, I., Cermak, M. (2015). Discretization and numerical realization of contact problems for elastic‐perfectly plastic bodies. PART I–discretization, limit analysis. {\em ZAMM‐Journal of Applied Mathematics and Mechanics/Zeitschrift für Angewandte Mathematik und Mechanik}, {95}(4), 333--353.

\bibitem[Sysala et al. (2016)]{SCKKZB16}
Sysala, S., Cermak, M., Koudelka, T., Kruis, J., Zeman, J., \& Blaheta, R. (2016). Subdifferential‐based implicit return‐mapping operators in computational plasticity. ZAMM‐Journal of Applied Mathematics and Mechanics/Zeitschrift für Angewandte Mathematik und Mechanik, 96(11), 1318--1338.

\bibitem[Temam (1985)]{T85}
Temam, R. (1985). {\em Mathematical Problems in Plasticity}.
Gauthier-Villars, Paris.

\bibitem[Tschuchnigg et al. (2015a)]{TSSLR2015}
Tschuchnigg, F., Schweiger, H.F., Sloan, S.W., Lyamin, A.V., \& Raissakis, I. (2015). Comparison of finite-element limit analysis and strength reduction techniques. Géotechnique, 65(4), 249--257.

\bibitem[Tschuchnigg et al. (2015b)]{TSS2015b}
Tschuchnigg, F., Schweiger, H.F., \& Sloan, S.W. (2015). Slope stability analysis by means of finite element limit analysis and finite element strength reduction techniques. Part I: Numerical studies considering non-associated plasticity. Computers and Geotechnics, 70, 169--177.

\bibitem[Tschuchnigg et al. (2015c)]{TSS2015c}
Tschuchnigg, F., Schweiger, H.F., \& Sloan, S.W. (2015). Slope stability analysis by means of finite element limit analysis and finite element strength reduction techniques. Part II: Back analyses of a case history. Computers and Geotechnics, 70, 178--189.

\bibitem[Vermeer and De Borst (1984)]{VB84}
Vermeer, P.A., \& De Borst, R. (1984). Non-associated plasticity for soils, concrete and rock. HERON, 29 (3), 1984.

\bibitem[Yu (2006)]{Y06}
Yu, H.-S. (2006). Plasticity and Geotechnics, Springer Science+Bussiness Media, New York.

\bibitem[Zienkiewicz et al. (1975)]{ZHL1975}
Zienkiewicz, O.C., Humpheson, C., and Lewis, R.W. (1975). Associated and non-associated visco-plasticity and plasticity in soil mechanics. G\'eotechnique, 25(4), 671--689.
}


\end{thebibliography}
\end{document}